\documentclass[11pt]{article}
\usepackage{fullpage}
\usepackage[bookmarks]{hyperref}
 \usepackage{amssymb}
 \usepackage{amsmath}
 \usepackage{amsthm}
 \usepackage{graphicx,color,colordvi}
 \usepackage{bbm}
 \usepackage{cleveref}
 \usepackage{stmaryrd}
\usepackage[utf8]{inputenc}
\usepackage{fullpage}
\usepackage[blocks]{authblk}
\usepackage{mathtools}

 \usepackage{pgfplots}

 \def\openone{\leavevmode\hbox{\small1\kern-3.8pt\normalsize1}}

 \def\dom{\operatorname{dom}}
 
 \def\CC{\mathbb{C}}
 \def\RR{\mathbb{R}}

 \def\NN{\mathbb{N}}
 \def\LL{\mathbb{L}}
 
 \def\11{\mathbf{1}}
 \def\LL{\mathcal{L}}

 \newtheorem{theorem}{Theorem}
 \newtheorem{lemma}{Lemma}

 \theoremstyle{definition}
 
 \newtheorem{example}{Example}

 \def\reff#1{(\ref{#1})}
 \def\eps{\mathcal{E}}
 
 \newcommand{\supp}{\mathop{\rm supp}\nolimits}
 
 \newcommand{\tr}{\mathop{\rm Tr}\nolimits}

 \newcommand{\spec}{{\rm sp}}

 \newcommand{\bra}[1]{\langle#1|}
 \newcommand{\ket}[1]{|#1\rangle}

 \newcommand{\cA}{{\cal A}}
 \newcommand{\cB}{{\cal B}}
 \newcommand{\cD}{{\cal D}}
 \newcommand{\cE}{{\cal E}}
 \newcommand{\cF}{{\chi}}
\newcommand{\cFF}{{\cal F}}

 \newcommand{\cT}{{\cal T}}
 \newcommand{\cH}{{\cal H}}
 
\newcommand{\cJ}{{\cal J}}

 \newcommand{\cP}{\mathcal{P}}
 
 \newcommand{\cL}{{\cal L}}

 \newcommand{\cZ}{{\cal Z}}

 \def\e{\mathrm{e}}
 
 \usepackage{graphicx}
 \usepackage{setspace}
 \usepackage{verbatim}
 \usepackage{subfig}

 \theoremstyle{definition}
 
 \theoremstyle{remark}

 \numberwithin{equation}{section}

 \DeclareRobustCommand\openone{\leavevmode\hbox{\small1\normalsize\kern-.33em1}}
 
 \newcommand{\id}{\rm{id}}
 \newcommand{\be}{\begin{equation}}
 	\newcommand{\ee}{\end{equation}}
 \newcommand{\bea}{\begin{eqnarray}}
 	\newcommand{\eea}{\end{eqnarray}}
 \newcommand{\beas}{\begin{eqnarray*}}
 	\newcommand{\eeas}{\end{eqnarray*}}

\begin{document}
 	\bibliographystyle{abbrv}
 	
 	\title{Contractivity properties of a quantum diffusion semigroup}
 	\author[1]{Nilanjana Datta}
 	\author[2]{Yan Pautrat}	
  	\author[1]{Cambyse Rouz\'e} 
  	\affil[1]{\small Statistical Laboratory, Centre for Mathematical Sciences, University of Cambridge, Cambridge~CB30WB, UK} 	
  	\affil[2]{\small Laboratoire de Math\'ematiques d'Orsay,
  		Univ. Paris-Sud, CNRS, Universit\'e
  		Paris-Saclay,  91405~Orsay, France}

 	\maketitle			 
\begin{abstract}
We consider a quantum generalization of the classical heat equation,
and study contractivity properties of its associated semigroup. We  prove
a Nash inequality and a logarithmic Sobolev inequality. The former
leads to an ultracontractivity result. This in turn implies that the largest eigenvalue and the purity of a state with positive Wigner function, evolving under the action of the semigroup, decrease at least inverse polynomially in time, while its entropy increases at least logarithmically in time.

\end{abstract}
\section{Introduction}

Interactions of an open quantum system with its environment often result in dissipative processes. The continuous time evolution of such
a system 
is hence in general non-unitary, unlike that of a closed system. In the time-homogeneous case, it is described by the quantum master equation
\begin{align}\label{master}
\frac{d\rho_t}{dt}&= \cL (\rho_t),
\end{align}
where $\rho_t$ denotes the state of the quantum system at time $t$,
and $\cL$ is the (infinitesimal) generator of the dynamics.
The solution of \Cref{master} is given by a quantum dynamical semigroup, which is a one-parameter family $\left(\Lambda_t\right)_{t \geq 0}$ of quantum operations satisfying certain conditions (see \Cref{sec-QDS}). 
 
An important class of results concerning quantum
dynamical semigroups and their classical counterparts (i.e., classical Markov semigroups) consists of certain
functional inequalities, e.g.~Nash inequality, logarithmic-Sobolev (or log-Sobolev) inequality, etc.
These inequalities can be used to find bounds on the time scales of various dissipative processes,
both in the infinite- and finite-dimensional setting: for example, in bounding mixing times of classical and quantum Markov processes\footnote{That is, the time it takes for the Markov process to become close to its stationary state, starting from any initial state.} (see e.g.~\cite{KT12} and references therein), estimating the rate at which the purity of a state decreases under the action of a quantum dynamical semigroup, estimating run times of quantum algorithms based on quantum Markov processes (e.g.~\cite{TOVPV11}), etc.
Functional inequalities are related to general formulations of certain convergence properties of the underlying semigroup. {\em{Hypercontractivity}} \cite{Nelson} is one such form of convergence and, for a semigroup which 
has an invariant state, it is related to exponential decay towards this state. The log-Sobolev inequality \cite{Gross} can be viewed as an infinitesimal version of hypercontractivity. Recently, Kastoryano and Temme \cite{KT15} defined non-commutative versions of another type of convergence called {\em{ultracontractivity}}, where the convergence towards an invariant state is at least inverse polynomial with respect to time. Ultracontractivity, in the classical case, can be derived for Markov semigroups satisfying a Nash inequality \cite{Nash}. Similarly, it was shown in \cite{KT15} that quantum (or non-commutative) ultracontractivity holds for quantum dynamical semigroups which satisfy a quantum version of the Nash inequality. For a brief summary of the functional inequalities and the contractivity properties of the associated semigroups, in the classical case, see \Cref{sec-ineqs}. Contractivity properties of quantum dynamical semigroups have been the focus of active research in recent years. See for example \cite{BK2016, M2012, KT12, Carbone2008, MFW16} and references therein.

 Classically, perhaps the oldest studied example of a
dissipative process is that of heat conduction. It is described by the simplest type of diffusion equation
\begin{align}\label{Heat-cl}
\partial_t u &= \Delta u,
\end{align} 
where $u\equiv u(x,t)$ is the temperature (or more generally, a scalar field) at the
position $x$, at the time $t$. In this paper we consider a quantum analogue of this equation and its associated
semigroup.
In this case the underlying Hilbert space of the quantum system is an infinite-dimensional separable Hilbert space and the unbounded generator
is given by\footnote{See \Cref{sec-Qdiff} for a more rigorous and complete definition.}
\begin{align}\label{HeatL}
\cL\left(\cdot \right) := - \frac{1}{4} \sum_{j=1}^n \left([Q_j,[Q_j, \cdot]] + [P_j,[P_j, \cdot]] \right),
\end{align}
where $Q_j, P_j, j=1,2, \ldots, n$ denote $n$ pairs of canonically conjugate observables satisfying the Heisenberg canonical commutation relations (CCR)
\begin{align}
[Q_j, P_k] = i\delta_{jk} I \, ; \, [Q_j, Q_k] = 0 =[P_j, P_k] \quad \forall \, j,k = 1,2, \ldots n.
\end{align}
In the above, we have set $\hbar=1$. The observables $Q_j, P_j, j=1,2, \ldots, n$ could be, for example, the ``position'' and ``momentum'' quadratures, respectively, of an $n$-mode bosonic field. The quantum master equation corresponding to the above generator (in this explicit form) was studied by Hall \cite{H94,H00} and later K\"{o}nig and Smith \cite{Koenig2012}. They showed that it can be viewed as a quantum analogue of the classical heat equation (\ref{Heat-cl}), and referred to it as the {\em{quantum diffusion equation}}. We adopt their nomenclature and refer to the associated semigroup as the {\em{quantum diffusion semigroup}} (see also \cite{L79} for more general ``diffusion'' semigroups).

We prove a Nash inequality and a log-Sobolev inequality for states evolving under the action of the semigroup. The former inequality in turn allows us to prove an ultracontractivity property of the quantum diffusion semigroup, which tells us that the purity of a state, as well as its largest eigenvalue, decrease at least inverse polynomially in time under the action of this semigroup, while its entropy increases logarithmically in time.  The above results, however, do not yield any information about mixing times, since our semigroup does not have an invariant state. 

In \Cref{sec-math} we give the necessary notations and definitions of mathematical objects arising in our work. In \Cref{sec-CCR} we introduce the algebra of canonical commutation relations (CCR). An important class of states of continuous variable quantum systems, namely, Gaussian states, to which some of our results pertain, are discussed in \Cref{sec-Gaussian}, and the more general Schwartz states are defined in \Cref{Schwartz}. Certain key entropic quantities which arise in the statement of our results are given in \Cref{sec-entropies}. \Cref{sec-QDS} contains a brief summary of the basics of quantum dynamical semigroups. The quantum diffusion semigroup and its properties are given in \Cref{sec-Qdiff}. Our main results are stated as six theorems in \Cref{sec-main-results}. The proofs of these are given in the following sections. Proofs of some of the intermediate results are given in the appendices.

\section{Notations and Definitions}
\subsection{Operators and Norms}\label{sec-math}
Let $\cH:=L^2(\RR^d)$ denote the infinite-dimensional separable Hilbert space of square-integrable functions on $\RR^d$, and let $\mathcal{B}(\cH)$ denote the set of linear bounded operators on $\cH$. Further, let $\cP(\cH)$ be the set of positive semi-definite operators on $\cH$, $\cP_{+}(\cH) \subset  \cP(\cH)$ be the set of (strictly) positive operators, and let $\cD(\cH):=\lbrace\rho\in\cP(\cH)\mid \tr\rho=1\rbrace$ denote the set of density matrices (or states) on $\cH$. We denote the support of an operator $A$ as ${\mathrm{supp}}(A)$, and its domain and kernel as 
${\mathrm{dom}}(A)$ and ${\mathrm{ker}}(A)$, respectively. Let ${\mathrm{ran}}(P)$ denote the range of a projection operator $P$, let $I\in\cP(\cH)$ denote the identity operator on $\cH$, and $\id:\cB(\cH)\to \cB(\cH)$ the identity map on operators on~$\cH$.

Let $\cT_1(\cH)$ denote the space of trace class operators equipped with the norm $||\sigma||_1 = \tr|\sigma|$, and let $\mathcal{T}_2(\cH)$ denote the space of Hilbert-Schmidt operators, i.e.,
\begin{align}
	\cT_2(\cH):=\{ A\in \cB(\cH): \tr|A|^2<\infty\}.
\end{align} 
This space is equipped with the Hilbert Schmidt norm 
\begin{align}\label{norm2}
\|A\|_2:=\left(\tr|A|^2\right)^{1/2}, 
\end{align}
associated to the Hilbert Schmidt inner product $\langle A,B\rangle:=\tr (A^*B)$, where by $A^*$ we denote the Hermitian conjugate of $A$. For any state $\rho\in\cD(\cH)$, $\tr \rho^2 \equiv \|\rho\|_2^2$
	is a measure of its purity. For any $p \in [1, \infty)$ we define the $p$-norm 
\begin{align}
\| A\|_p := \left(\tr(|A|^p)\right)^{1/p},
\end{align}
and for $p = \infty$,
\begin{align}
\| A\|_\infty := \sup_{\|x\|=1} \| Ax\|.
\end{align}
For a given state $\rho\in\mathcal{D}(\mathcal{H})$, for any $p \geq 1$, we similarly define $L^p(\rho)$ to be the Banach space of operators with norm
\begin{align}
	\|A\|_{p,\rho}:=\tr(\rho|A|^p)^{1/p}.
	\end{align}
A quantum operation (or quantum channel) is given by a linear, completely positive trace-preserving (CPTP) map (see e.g.~\cite{NC}). For a linear map $\Lambda: \cB(\cH) \rightarrow \cB(\cH)$, and any $p,q \in  [1, \infty]$ we define the following norm:
\begin{align}\label{pq-op}
\|\Lambda\|_{p \to q} := \sup_{A \in \cB(\cH):\atop{\| A\|_{p} = 1}} \| \Lambda (A)\|_q
\end{align}
Similarly, we denote by $L^p(\mu)$ the space of $p$-integrable functions on ${\mathbb{R}}^d$, with respect to a probability measure $\mu$, with associated norm
\begin{align}
\|f\|_{p,\mu} &:= \left( \int |f(x)|^p d\mu(x) \right)^{\frac{1}{p}}.
\end{align}
We denote these simply by $L^p$, $\|\cdot\|_{p}$ when $\mu$ is the Lebesgue measure. We define the norm
\begin{align}
	\|\Lambda\|_{L^p \to L^q} := \sup_{f \in L^p:\atop{\| f\|_{p} = 1}} \| \Lambda (f)\|_{q}.
	\end{align}
\subsection{The algebra of canonical commutation relations (CCR)}\label{sec-CCR}
Given the canonically conjugate observables $Q_j, P_j, j=1,2, \ldots, n$, let us introduce a $2n$-dimensional column vector of observables 
\begin{align}\label{r-eq}
R := (Q_1, P_1, \ldots, Q_n, P_n)^T,
\end{align}
where the superscript $T$ denotes a transpose.
Then the canonical commutation relations (CCR) can be compactly expressed as 
\begin{align}
[R_j, R_k] = i \Omega_{jk} 
\end{align}
where
\begin{align}\label{Omega}
   \Omega&=
  \left[ {\begin{array}{cc}
   0 & 1 \\       {-1} & 0 \      \end{array} } \right]^{\oplus n}.
\end{align}
A rigorous definition of the operators $Q_j, P_j$ is as follows. Consider a family $\{V(z) : z=(q_1, p_1, \ldots, q_n, p_n)\in\RR^{2n}\}$ of unitary operators satisfying the so-called Weyl-Segal CCR \cite{Holevo2011},
\begin{align}\label{CCRrelation}
V(z)V(z') &= e^{-\frac{i}{2}\{z,z'\}} V(z + z')
\end{align}
where  
\begin{align}
\{z, z'\}&:= (z, \Omega z') = \sum_{j=1}^n (q_j p'_j - q'_j p_j)
\end{align}
is the canonical symplectic form, and the notation $(z,z'')$ denotes the scalar product of the vectors $z$ and $z''$ in $\RR^{2n}$. 
 The vector space $\RR^{2n}$ equipped with the symplectic form $\{z, z'\}$ defined above, constitutes a symplectic vector space which we denote as $\cZ$. It can be viewed as the {\em{quantum phase space}}, e.g.~the phase space of an $n$-mode bosonic field. 

Stone's Theorem (see e.g.~\cite{Reedsimon}) allows us to define $P_j$ (resp.~$Q_j$) as the self-adjoint generator of $\left(V(tq^{(j)})\right)_{t\in {\mathbb{R}}}$ (resp.~$\left(V(-tp^{(j)})\right)_{t\in {\mathbb{R}}}$) where $q^{(j)}$ (resp.~$p^{(j)}$) is the vector $z$ with $q_j=1$ (resp.~$p_j=1$) and all other elements equal to zero.  The operators $V(z)$ can then be expressed as follows: 
\begin{align}\label{Weyl}
V(z) &:= e^{i(z, \Omega R)}\equiv \exp[i \sum_{j=1}^n ( q_jP_j-p_jQ_j )].
\end{align}
The operators $V(z)$  are called the Weyl operators or Weyl displacement operators, the latter name being justified by the relation
\begin{align}
V(z)R_j V(-z) &= R_j + z_j I \quad \forall \, j \in \{1,2, \ldots, n\}.\label{commRV}
\end{align}

A quantum state $\rho \in \cD(\cH)$ of the system is uniquely defined on $\cZ$ through its {\em{quantum characteristic function}}
\begin{align}\label{qchar}
\cF_\rho(z) := \tr\left(\rho V(z)\right).
\end{align}
\textcolor{black}{The inverse Fourier transform of the characteristic function
	\begin{align}\label{wign}
	\mathcal{F}^{-1}_{\chi_\rho}(u):={(2\pi)^{-n}}\int_{\RR^{2n}}\e^{-i(u,z)}\chi_\rho(z)dz
	\end{align}
	is called the \textit{Wigner function} of the state\footnote{Note that the above definition differs from the usual one by a factor of $(2\pi)^{-n}$.}. Under this convention the integral of the Wigner function is equal to $(2\pi)^n$. However, they can in general be negative, and therefore cannot be interpreted as probability densities associated to states.} \\\\
A quantum state is said to have finite $k^{\rm{th}}$ moment, for $k \in {\mathbb{N}}$ if 
$$ \tr \big(\rho |R_j|^k\big) < \infty, \quad \forall \, j \in \{1,2, \ldots, 2n\}, $$
where $R_j$ denotes the $j^{th}$ entry of the vector $R$ of observables defined through \reff{r-eq}.

We employ the following theorems to prove our main results.
\begin{theorem}[Noncommutative Parseval's relation, \cite{Holevo2011}\label{NCparseval} Theorem 5.3.3] The function given by (\ref{qchar}) 
extends uniquely to a map
$T\mapsto \cF_T$  from the Hilbert space $\cT_2(\cH)$ onto $L^2(\RR^{2n})$, so that  
	\begin{align}
		\tr(T_1^*T_2)=(2\pi)^{-n}\int_{\RR^{2n}} \overline{\cF_{T_1}(z)}\cF_{T_2}(z) dz,
		\end{align}
(where the notation $ \overline{f(z)}$ denotes the complex conjugate of $ f(z)$).	\end{theorem}
	
 We also employ the following lemma, which is an adaptation of a result proved by Holevo
in \cite{Holevo2011}.
 \begin{lemma}[See \cite{Holevo2011} Lemma 5.4.2]\label{lemma5.4.2} Let $X$ be a (possibly unbounded) self-adjoint operator and let $\rho$ be a state such that 
 	\begin{align}
 		\tr(\rho X^2)<\infty.
 		\end{align}
 	Then, 
 	\begin{align}
 		X=i^{-1}\left.\frac{d}{dt}\right|_{t=0}U_t
 	\end{align}	
 	where $U_t:=\e^{it X}$ is defined by functional calculus, and the derivative is taken in $L^2(\rho)$, i.e., 
 		\begin{align}
 		\tr\left[\rho \left|  \frac{U_t-I}{t}-iX   \right|^2\right]\underset{t\to 0}{\rightarrow}  0.
 		\end{align}
 \end{lemma}	
\subsection{Gaussian states}\label{sec-Gaussian}
An important class of states for a continuous variable quantum system are the Gaussian states. They play a fundamental role in the study of various quantum information processing tasks involving continuous variable quantum systems which are relevant in quantum computation and quantum cryptography. Important examples of Gaussian states arise in quantum optics and include coherent states, squeezed states and thermal states. A Gaussian state $\rho$ is defined via its characteristic function
\begin{align}\label{char-fn-G}
	\cF_\rho(z)=\e^{i(\mu, z)-\frac{1}{2}(z, \Sigma z)}\qquad z\in\cZ,
\end{align}
where $\mu$ is the mean vector of the state $\rho$, with elements
\begin{align}
	\mu_i &:=\tr(\rho R_i), \quad i=1,2, \ldots, 2n,
\end{align}
and $\Sigma$ is its $2n \times 2n$ (real, symmetric) covariance matrix, with elements
	\begin{align}
\Sigma_{ij}&:= \frac{1}{2}\tr(\rho(R_i^cR_j^c+R_j^cR_i^c)),
\end{align} 
where $R_i^c:=R_i-\tr(\rho R_i)$. The characteristic function given by \Cref{char-fn-G} is exactly the one of a $2n$-dimensional Gaussian vector of mean $\mu$ and covariance $\Sigma$. The Gaussian state is said to be {\em{centered}} if $\mu=0$. Gaussian states have some very useful properties (see \cite{DMGH15,Holevo2012}). For example, the linear span of all Gaussian states is dense in the Banach space, $\cT_1(\cH)$, of trace class operators. Moreover, any Gaussian state can be written in a rather simple way as follows.	
\begin{theorem}[\cite{Holevo2012} Lemma 12.21 and Theorem 12.22]\label{theoremgaussstate}
	A Gaussian state $\rho$ with covariance matrix $\Sigma$ is invertible if and only if 
	\begin{align}
		\det\left( \Sigma +\frac{i}{2}\Omega\right) \ne 0.
	\end{align}	
	Moreover, any Gaussian state $\rho\equiv \rho_{\mu,\Sigma}$ of mean $\mu$ and covariance matrix $\Sigma$ can be written as
	\begin{align}\label{gausstogausscentered}
		\rho_{\mu,\Sigma}=V(\mu)\rho_{0,\Sigma}V(-\mu),
	\end{align}	
	where $\rho_{0,\Sigma}$ is a centered Gaussian state of covariance matrix $\Sigma$, and $V(\mu)$ denotes the Weyl operator given by \Cref{Weyl}. Further, a centered, invertible Gaussian state can be expressed as follows:
\begin{align}
	\rho_{0,\Sigma}=C\e^{-R^T \Gamma R},\label{exprrho}
\end{align}	
where 
\begin{align}
C &=\left[\det\left(\Sigma+\frac{i}{2}\Omega\right) \right] ^{-1/2},
\end{align}
with $\Omega$ given by \Cref{Omega} and $\Gamma$ defined through the relation 
\begin{align}\label{gamma}
2 \Omega^{-1}\Sigma=\cot(\Gamma\Omega).
\end{align}
\end{theorem}

\subsection{Schwartz operators} \label{Schwartz}
Recently Keyl, Kiukas and Werner introduced a larger class of operators, the so-called Schwartz operators, in \cite{KKW16}. We first recall that a Schwartz function on $\RR^n$ is a smooth function $\varphi:\RR^n\mapsto \CC$ for which
\begin{align}
\sup_{x\in\RR^n}\big|x_1^{\alpha_1}\dots x_n^{\alpha_n} \frac{\partial^{\beta_1}}{\partial x_1^{\beta_1}}\dots\frac{\partial^{\beta_n}}{\partial x_n^{\beta_n}} \varphi(x)\big|<\infty
\end{align}
for all $\alpha,\beta\in\NN^n$. The set of Schwartz functions is denoted by $\mathcal{S}(\RR^n)$. A Hilbert Schmidt operator $T$ is called a \textit{Schwartz operator} if its characteristic function $\chi_T$ is in $\mathcal{S}(\RR^{2n})$. Note that any Gaussian state is a Schwartz operator. More generally, a quantum state which is a Schwartz operator is called a \textit{Schwartz state}.

\subsection{Entropic quantities}\label{sec-entropies}
The von Neumann entropy of a quantum state $\rho \in \cD(\cH)$ is given by 
\begin{align}
	S(\rho):= -\tr(\rho\ln\rho).
\end{align}
The quantum relative entropy of two quantum states $\rho,\sigma \in \cD(\cH)$ is given by:
\begin{align}
	D(\rho||\sigma):=\left\{ \begin{aligned}
	&\tr(\rho(\ln\rho-\ln\sigma))\quad {\hbox{if }} \operatorname{supp}\rho\subset \operatorname{supp}\sigma\nonumber \\
		&\infty\quad {\hbox{else.}}
		\end{aligned}\right.   
\end{align} 
 Consider a family $\{\rho^{(\theta)}\}_{\theta\in\RR}$ of quantum states on some Hilbert space $\cH$, parametrized by a real number $\theta$, such that $\rho^{(0)}=\rho$, and such that the function $\theta\mapsto D(\rho||\rho^{(\theta)})$ is twice differentiable at $0$. We consider the following entropic quantity of the family of $\{\rho^{(\theta)}\}_{\theta\in\RR}$, called the {\em{divergence-based Fisher information}} \cite{Koenig2012,Koenig2015}:
	\begin{align}\label{Fisherinfo}
		\cJ(\{\rho^{(\theta)}\})&:=\left.\frac{d^2}{d\theta^2}D(\rho||\rho^{(\theta)})\right|_{\theta=0}.
	\end{align}
In \Cref{concavity}, we define a particular family of quantum states (see \Cref{shift-states}) and study the divergence-based Fisher information associated to it.

We also consider the following entropic quantity, $J(\rho)$ of a state $\rho$, which we call the {\em{entropy variation rate}} (the nomenclature is justified by \Cref{quantumdebruijin}):
\begin{align}\label{entropy-var-rate}
		J(\rho)&:=\sum_{j=1}^{2n}\left.\frac{d^2}{d\theta^2}D(\rho||
\rho^{(\theta)}_{R_j})\right|_{\theta=0} \equiv \sum_{j=1}^{2n}\left.\cJ(\{\rho^{(\theta)}_{R_j}\})\right|_{\theta=0},
	\end{align}
where 
\begin{align}\label{shift-states}
\rho^{(\theta)}_{R_j}&:= e^{+i\theta R_j} \rho \, e^{-i\theta R_j}, \quad j \in \{1,2, \ldots 2n\}
\end{align} 
where $R_j$ denotes the $j^{th}$ element of the vector $R$ given by \Cref{r-eq}.
\medskip

\noindent
We use yet another entropic quantity in our analysis, namely the {\em{quantum entropy power}} of an $n$-mode state $\rho$ (as defined in \Cref{sec-CCR}), which is defined as
follows:
\begin{align}\label{EP}
	E(\rho):=\e^{S(\rho)/n}.
\end{align}

\subsection{Quantum Dynamical Semigroups}\label{sec-QDS}

In this section, we briefly recall the basics of quantum dynamical semigroups. $\cL$ is a linear operator defined on its domain ${\rm{dom}}(\cL) \subset \cT_1(\cH)$. The solution of \reff{master}, under the assumption of
Markovianity, is given by a one-parameter family $\left(\Lambda_t\right)_{t \geq 0}$ of quantum operations, i.e. linear, completely positive, trace-preserving maps on $\mathcal T _1(\mathcal H)$ satisfying the following properties
\begin{itemize}
\item $\Lambda_0 = {\rm{id}}$, where ${\rm{id}}$ denotes the identity map;
\item $\Lambda_t \Lambda_s = \Lambda_{t+s}$ $\,-\,$ semigroup property;
\item $\underset{t\to 0}{\lim} || \Lambda_t(\rho) - \rho||_1 = 0$ $\,-\,$ strong continuity.
\end{itemize}
The parameter $t$ plays the role of time. Hence, $\Lambda_t$ results in time evolution over the interval $[0,t]$, and is called a {\em{quantum dynamical map}}. Accordingly, the semigroup $\left(\Lambda_t\right)_{t \geq 0}$ is called a {\em{quantum dynamical semigroup}} (QDS) or a quantum Markov semigroup. Formally, one writes $\Lambda_t = e^{t\cL}$, and refers to $\cL$ as the generator (or infinitesimal generator) of the semigroup. It is also called the Lindblad superoperator or Liouvillean. The latter name stems from the fact that it is a generalization of the superoperator arising in the Liouville-von Neumann equation, which governs the unitary time evolution of the state of a closed quantum system. The semigroup property embodies the assumptions of time-homogeneity and Markovianity, since it implies
that the time evolution is independent of its history and of the actual time. The quantum master equation can be expressed as
\begin{align}
\label{evo2}
\frac{d \Lambda_t(\rho)}{dt} = \cL \circ \Lambda_t(\rho).
\end{align}
Henceforth we denote $\rho_t= \Lambda_t(\rho)$.
 The QDS $\left(\Lambda_t\right)_{t \geq 0}$ gives a description of the dynamics (of states) in the Schr\"odinger picture. In the Heisenberg picture, the 
dynamics (of observables) is given by the semigroup $(\Lambda^*_t)_{t\ge 0}$ of linear, completely positive, unital maps  on $\cB(\cH)$, satisfying the following properties:
\begin{itemize}
	\item $\Lambda^*_0 = {\rm{id}}$, where ${\rm{id}}$ denotes the identity map;
	\item $\Lambda^*_t \Lambda^*_s = \Lambda^*_{t+s}$ $\,-\,$ semigroup property;
		\item $\lim_{t \to 0} || \Lambda^*_t(X) - X||_\infty = 0$ $\,-\,$ strong continuity.
\end{itemize}
Note that $\Lambda_t^*$ is the adjoint of $\Lambda_t$, i.e., $\tr(A \Lambda_t(\rho)) = \tr(\Lambda_t^*(A)\rho)$
for any $A \in {\mathcal{B}}({\mathcal{H}})$, and any $\rho\in \cD(\cH)$. Since
$\Lambda_t^*$ is positive and unital, it satisfies the following contractivity property: $\|\Lambda_t^*\|_\infty\le 1$. The semigroup in the Schr\"odinger picture is also (trivially) contractive, in fact, $\| \Lambda_t \|_1 = 1$, since $\Lambda_t$ is trace-preserving.

 A state $\rho$ is said to be an {\em{invariant state}} with respect to the semigroup if $\Lambda_t(\rho)=\rho$, or equivalently
 \begin{align}
 	\tr(\rho \Lambda_t^*(A))=	\tr(\rho A),\qquad \forall \, A\in\mathcal{B}(\cH), \,\, \forall t \geq 0.
 	\end{align}

\section{Quantum Diffusion Semigroup}\label{sec-Qdiff}

In this paper we consider a {\em{quantum diffusion equation}}, which is a quantum master equation with generator $\LL$ given by 
\begin{align}\label{HeatL}
\cL\left(\cdot \right) := - \frac{1}{4} \sum_{j=1}^{2n} [R_j,[R_j, \cdot]],
\end{align}
on $$D:= {\rm{span}}\bigl\{ \ket{\varphi}\bra{\psi} \,:\, \ket{\varphi}, \ket{\psi} \in \bigcap_{j=1}^n \left({\rm{dom}} (P_j^2) \cap {\rm{dom}} (Q_j^2)\right)\bigr\},$$
with $R_j$ being the $j^{th}$ element of the vector $R$ defined through \Cref{r-eq}, and
\begin{align}
{\rm{dom}} (Q_j^2) &= \left\{ \psi \in L^2( {\mathbb{R}}^n) \,:\, \int_{{\mathbb{R}}^n} |q_j^2 \,\psi(q)|^2 dq < \infty \right\},\nonumber\\
{\rm{dom}} (P_j^2) &= \left\{ \psi \in L^2( {\mathbb{R}}^n) \cap C^2({\mathbb{R}}^n)\,:\, \int_{{\mathbb{R}}^n} |\partial_j^2 \psi(q)|^2 dq < \infty \right\}.
\end{align}

By Section 2 of \cite{Holevo1996}, $\cL$ extends to an unbounded operator on $\mathcal T_1(\mathcal H)$ such that, if an operator $X$ has finite moments of order 2, then $X\in \mathrm{dom}(\mathcal L)$; this in particular is true if $X$ is a Gaussian state.

The following theorem lists some important properties of the semigroup. Some of these were proved by
K\"{o}nig and Smith in \cite{Koenig2012}.

\begin{theorem} \label{theo_KS}
Consider the quantum diffusion semigroup $(\Lambda_t)_{t \geq 0}$ with generator $\cL$ defined through \Cref{HeatL}.
 	\begin{enumerate}
 		\item{The semigroup is reversible (or symmetric). That is, for all $A,B \in \dom(\LL)$, 
 		\begin{align}
 		\langle A, \LL(B)\rangle=\langle \LL(A), B\rangle.\label{Lhermitian}
 		\end{align}}
 		\item{Let the state $\rho_t\equiv \Lambda_t(\rho)$ denote the solution of the quantum diffusion equation defined through \Cref{evo2,HeatL}, when the initial state is $\rho$. Then for each $t\ge 0$, the characteristic function of $\rho_t$ 
is given by
 		\begin{align}\label{rhorhotcharac}
 			\cF_{\rho_t}(z)=\cF_{\rho}(z)\e^{-|z|^2t/4}, \quad \forall \, z \in \cZ,
 		\end{align}
 and we have
\begin{align}\label{q-conv}
\equiv \rho_t \equiv \rho * g_{t/2} &:= \int_{\cZ} dz \, g_{t/2}(z) \cA(z)\left(\rho\right),
\end{align}
where $g_{t/2}$ denotes the probability density function (pdf) of a Gaussian random variable on ${\mathbb{R}}^{2n}$ with zero mean and variance equal to $t/2$, i.e.
\begin{align}\label{gt/2}
g_{t/2}(z) &= \frac{1}{(\pi t)^n} e^{- |z|^2/t},
\end{align} 
and $\cA(z)$ is the automorphism defined through the relation 
\begin{align}
\cA(z) (\rho) &:=  V(z) \rho \,V(z)^* = V(z) \rho\, V(-z),
\end{align}
where $V(z)$ denotes the Weyl operator given by \reff{Weyl}. In particular, if $\rho$ is a Gaussian state, then so is $\rho_t$.}
\item{The semigroup does not possess an invariant state.}
\item{The semigroup is {\em{self-dual}}, i.e.,
\begin{align}\label{self-dual}
\langle A, \Lambda_t(B)\rangle = \langle \Lambda_t(A), B \rangle,\quad \forall A, B \in \cT_2(\cH).
\end{align}}
 	\end{enumerate}
\end{theorem}
\begin{proof}
\Cref{Lhermitian} and \Cref{rhorhotcharac} were given in \cite{Koenig2012}. In fact, \Cref{Lhermitian} can be directly verified by using \Cref{HeatL}. The absence of an invariant state can be inferred from \Cref{rhorhotcharac} as follows: if there was an invariant state $\rho$, then its characteristic function would satisfy $\cF_{\rho_t}=\cF_{\rho}$ for all $z \in \cZ$. However, this is impossible by \Cref{rhorhotcharac}. The proof of \Cref{q-conv} is obtained as follows:
Note that $\e^{-|z|^2t/4}$ is the characteristic function of a Gaussian random variable of associated pdf $g_{t/2}$. It is well-known that if $f$ and $g$ are two pdfs then the characteristic function of their convolution $f * g$ is equal to the product of their characteristic functions:
$$  \chi_{f * g} (z) = \chi_f(z) \chi_g (z).$$
This
property also holds when one of the pdfs is replaced by a quantum state and the standard definition of convolution is replaced by the one of {\em{quantum convolution}} (given by \Cref{q-conv}), as shown by Werner \cite{Werner1984} (see also \cite{K72}) in his generalization of harmonic analysis to the quantum framework. Hence, \reff{rhorhotcharac} allows us to express the state $\rho_t$ as a quantum convolution of the initial state $\rho$ and the pdf $g_{t/2}$, as given by \Cref{q-conv}.

The proof of self-duality, \Cref{self-dual}, is as follows: the operator $\Lambda_t(B)$ is given by the right hand side of \Cref{q-conv} with $\rho$ replaced by $B$. Hence,  
\begin{align}
\tr \left(A^* \Lambda_t(B)\right) &= \tr \left(A^* \int_{\cZ}dz \,g_{t/2}(z) V(z) B V(-z) \right) \nonumber\\
&= \tr \left(A^* \int_{\cZ}dz\, g_{t/2}(z) V(-z) B V(z) \right) \nonumber\\
&=\tr \left(\int_{\cZ}dz \,g_{t/2}(z) V(z)A^*V(-z) B  \right) \nonumber\\
&= \tr \left(\left(\int_{\cZ}dz \,g_{t/2}(z) V(z)AV(-z)\right)^* B  \right) \nonumber\\
&= \tr \left( (\Lambda_t(A))^*B\right)
\end{align}
where we have used the symmetry of the Gaussian pdf, the cyclicity of the trace, the fact that $V(z)^* = V(-z)$, and \Cref{q-conv} with $\rho$ replaced by $A$.
\end{proof}
\smallskip

\noindent
{\em{Remark:}} \Cref{q-conv} establishes a direct analogy with the classical heat semigroup $f\ast g_{2t}$ on ${\mathbb{R}}^{n}$, which is the solution of the heat equation:
   \begin{align}\label{c-Heat}
   {\partial_t f}=\Delta f.
   \end{align}
   \medskip

We define the following quantity associated to our QDS:
\begin{align}\label{Dirichlet}
   	\mathcal{E}(\rho):=-\tr\left(\rho\LL(\rho)\right),\qquad \rho\in\dom(\LL).
   	\end{align}
It is called the {\textit{Dirichlet form}} (see e.g.~\cite{Davies1993}) and uniquely characterizes the QDS. Analogously, the Dirichlet form characterizing the semigroup associated with the classical heat equation, \Cref{c-Heat}, in the case in which $f$ is a smooth function on ${\mathbb{R}}^n$, for which $f \nabla f$ vanishes at infinity, is given by
\begin{align}
\cE_{cl}(f)&:= -\int_{{\mathbb{R}}^{2n}} f(x) \Delta f(x) dx \nonumber\\
& = \|\nabla f\|_{2}^2
\label{DF-cl-heat}
\end{align}
where the last line follows by a simple integration by parts.

We employ the following quantum version of de Bruijn's identity (Theorem 7.3 of \cite{Koenig2015}) in our proof. 
\begin{theorem}[Quantum de Bruijn's identity]\label{quantumdebruijin}
	Let $\rho$ be a centered, $n$-mode Gaussian state. Then
	\begin{align}\label{qbruijn}
		\frac{d}{dt}S(\rho_t)=\frac{1}{4}J(\rho_t),
	\end{align}	
where $\rho_t = \rho \ast g_{t/2}$ for all $t \geq 0$.
\end{theorem}	
 
\subsection{Nash inequality, log-Sobolev inequality and contractivity properties}\label{sec-ineqs}

The classical Nash inequality was introduced by Nash \cite{Nash} to obtain regularity properties on the solutions of parabolic partial differential equations.
In the Euclidean space ${\mathbb{R}}^n$, it can be stated as follows \cite{BBG11}: there exists a constant $c_n >0$ (depending only on $n$), such that for any real-valued, smooth function $f$ vanishing at infinity, 
\begin{align}\label{eq-Nash}
\|f\|_{2}^{1+n/2} &\leq c_n \|f\|_{1}\|\nabla f\|_{2}^{n/2}.
\end{align}
The optimal constant $c_n$ was later evaluated by Carlen and Loss \cite{CL93}. A similar inequality holds for complex-valued functions, with a modified constant.
Nash inequality implies {\em{ultracontractivity}} of the heat semigroup associated with \Cref{c-Heat}, which means that it maps $L^1(\RR^n)$ to $L^\infty(\RR^n)$ with
\begin{align}
	\|\Lambda_t\|_{L^1\to L^\infty}\le a_n(t):=\left( \frac{1}{\pi et}\right)^{n/2},
\end{align}	
 where $\Lambda_t(f)=f\ast g_{2t}$. The above inequalities may be stated in the general framework of symmetric Markov semigroups,
where it is a simple and powerful tool to study regularity properties of the underlying semigroup. 

The inequality \ref{eq-Nash} can be expressed in terms of the Dirichlet form, $\cE_{cl}(f)$,  of the classical heat equation as follows:
\begin{align}\label{eq-cl}
\|f\|_{2}^{2 + 4/n} & \leq c_n^{4/n} \|f\|_{1}^{4/n} \cE_{cl}(f).
\end{align}

The log-Sobolev inequality was first introduced by Gross, with applications to quantum field theory. He stated it as follows \cite{Gross}: let $\mu$ denote the standard Gaussian measure on ${\mathbb{R}}^n$, i.e.,
$$ d\mu(x) :=g(x)dx:= \frac{1}{(2\pi)^{n/2}} e^{-|x|^2/2} dx, \quad x \in {\mathbb{R}}^n$$
where $dx$ denotes the Lebesgue measure on ${\mathbb{R}}^n$. Then the log-Sobolev inequality is given by
\begin{align}\label{eq-log-Sob}
\int_{{\mathbb{R}}^n}|f(x)|^2 \ln |f(x)| d\mu(x)-||f||_{2,\mu} \ln ||f||_{2,\mu} & \leq
\int_{{\mathbb{R}}^n}|\nabla f(x)|^2 d\mu(x), 
\end{align}
where we use the notation $||f||_{p,\mu} := \left(\int |f(x)|^p d\mu(x)\right)^{1/p}$. 

In \cite{Gross1975} Gross showed that \Cref{eq-log-Sob} is equivalent to {\em{hypercontractivity}} of the Ornstein-Uhlenbeck semigroup $(\Lambda_t^{(OU)*})_{t\geq 0}$, the generator of which is given by
\begin{align}
	\partial_t f(x,t)=\LL_{OU}^*(f)(x,t):= \Delta f(x,t)+x.\nabla f(x,t),\quad x\in\RR^n, \, t\ge 0,
\end{align}
Hypercontractivity denotes the following property:
\begin{align}
	\| \Lambda_t^{(OU)*}\|_{L^p(\mu) \to L^q(\mu)}\le 1\qquad \forall\,p,q \in {\mathbb{N}} \quad {\hbox{such that}} \quad t\ge \frac{1}{2}\ln\left( \frac{q-1}{p-1}\right).
  \end{align}
\section{Main Results}\label{sec-main-results}
In this paper, we prove that the quantum diffusion semigroup, whose generator 
is given by \Cref{HeatL}, satisfies a Nash inequality. We also derive a log-Sobolev inequality in a form analogous to the ones given e.g.~in \cite{Chafai2005,Toscani}\footnote{These inequalities are known to be equivalent to the log-Sobolev inequality for the classical Ornstein-Uhlenbeck semigroup in its usual form given in \cite{Gross}.}. In the following we denote $\rho_t= \Lambda_t(\rho)$, where $\Lambda_t$ is the quantum diffusion semigroup. 
\begin{theorem}[A non-commutative Nash inequality]\label{thm:nash}
 \textcolor{black}{If $\rho$ is a Schwartz state of positive Wigner function, then $\rho\in\dom(\LL)$ and the following non-commutative Nash inequality holds: there exists a positive constant $C_n$ (depending only on $n$) such that}
 	\begin{align}\label{ineq:nash}
 		\|\rho\|_2^{2+2/n}\le  C_{n}\eps(\rho),
 	\end{align}
where $\eps$ is the Dirichlet form associated to the quantum diffusion semigroup (see \Cref{Dirichlet}).
 \end{theorem}	

\begin{theorem}[Non-commutative ultracontractivity] 
	\label{thm:ultra}	
	If $\rho$ is a Schwartz state of positive Wigner function, then there exists a positive, then there exists a positive constant $\kappa_n$, such that for any $t>0$
	\begin{align}
	&\|\rho_t\|_\infty\le\|\rho_t\|_{ 2}\le \kappa_n t^{-n/2}\label{UC1to2}.
	\end{align}
	
	Hence, for any such initial state $\rho$ evolving under the action of the semigroup $(\Lambda_t)_{t\ge 0}$, the following bounds give the rate of decay of its purity and the rate of increase of its von Neumann entropy, respectively:
	\begin{align}\label{purity}
	\tr \rho_t^2 \equiv \|\rho_t\|_2^2&\le \kappa_n^2 t^{-n};\\
	S(\rho_t) & \geq \frac{n}{2} \ln \left(\kappa_n^{-2/n}t\right).
	\end{align}
\end{theorem}	

The proofs of \Cref{thm:nash} and \Cref{thm:ultra} are given in \Cref{Nashproof}. 
\begin{theorem}[A non-commutative log-Sobolev inequality]\label{logsob}
For any invertible, centered, $n$-mode Gaussian state $\rho$,
\begin{align}\label{NC-lS}
\tr(\rho\ln \rho)+n\le \frac{ J(\rho)}{2\e},
\end{align}
where $J(\rho)$ is the entropy variation rate defined through \Cref{entropy-var-rate}.
\end{theorem}
\smallskip

\noindent
{\em{Remark:}} To see why \Cref{NC-lS} can be viewed as a non-commutative log-Sobolev inequality, let us first rewrite it as follows:
\begin{align}\label{NC-lS2}
J(\rho) \geq e(2n - 2S(\rho)). 
\end{align}
In the  classical case, Toscani proved the following inequality for any smooth function $f$ on ${\mathbb{R}}^n$ which vanishes at infinity (cf.~eq.(24) of \cite{Toscani}): for any $\sigma>0$:
\begin{align}\label{Tos1}
\sigma J_{cl}(f) /2&\geq n - H(f) + n/2 \ln (2\pi\sigma),
\end{align}
where 
\begin{align}
\label{cl-FI}
J_{cl}(f) &:= \int_{{\mathbb{R}}^n} \frac{|\nabla f(x)|^2}{f(x)} dx
\end{align}
is the classical Fisher information, and $H(f) := - \int_{{\mathbb{R}}^n} f(x) \ln f(x) dx$ is the Shannon differential entropy.
Let us choose $\sigma=1$ and $f=gh^2$, where $g$ denotes the probability density function of a standard normal distribution on ${\mathbb{R}}^n$, and $h$ is a function
for which $\int_{{\mathbb{R}}^n}|h(x)|^2 g(x) dx =1$. Then on evaluating the expressions for $J_{cl}(f)$ and $H(f)$ and substituting them in \Cref{Tos1}, the latter reduces to the standard log-Sobolev inequality of Gross, given by \Cref{eq-log-Sob}. For the choice $\sigma=1/(2\pi e)$, \Cref{Tos1} reduces to 
\begin{align}
	J_{cl}(f)/\pi \ge \e (2 n-2H(f)),
\end{align}	
which is completely analogous to \Cref{NC-lS2}. This analogy between our inequality, as given by \Cref{NC-lS2}, and \Cref{Tos1} is what leads us to refer to \Cref{NC-lS2} as a non-commutative log-Sobolev inequality.
\medskip

\noindent
\Cref{logsob} follows directly from the following quantum analogue of the classical isoperimetric inequality for entropies:
\begin{theorem}[Isoperimetric inequality for the quantum entropy\label{iso}\footnote{The nomenclature comes from the fact that \Cref{isope} is analogous to the classical isoperimetric inequality $J_{cl} (f)E(f)\ge 2\pi en$, where $E(f) = e^{2H(f)/n}$.}]
	For any centered, $n$-mode Gaussian state $\rho$,
		\begin{align}\label{isope}
			J(\rho) E(\rho)\ge 2\e n,
		\end{align}
where $ E(\rho)$ is the entropy power of the state $\rho$, defined through \Cref{EP}.
	\end{theorem}

In order to prove \Cref{iso}, we first prove the following intermediate result, which might be of independent interest:
\begin{theorem}[Quantum Blachman-Stam inequality]\label{blachmanstam}
	For any $\alpha, \beta >0$, $t >0$, for any state $\rho$ such that $\theta\mapsto D(\rho||\rho_{R_j}^{(\theta)})$ is twice differentiable at $0$ for all $j=1,...,2n$, we have
	\begin{align}\label{qfisherinfoin}
		(\alpha+\beta)^2J(\rho_t)\le \alpha^2 J(\rho) + \frac{4n \beta^2}{t} ,
	\end{align}
where $\rho_t = \Lambda_t(\rho)$, with $\left(\Lambda_t\right)_{t \geq 0}$ being the quantum diffusion semigroup.
\end{theorem}
The above theorem leads to a result about the concavity of the entropy power with respect to time ($t$), which is given as follows.
\begin{theorem}[Concavity of the quantum entropy power]\label{concavqep}
	For any centered, invertible $n$-mode Gaussian state $\rho$, the entropy power $E(\rho_t)$ defined through \Cref{EP} is twice differentiable as a function of time ($t$), and for all $t\ge 0$:
	\begin{align}
		\frac{d^2}{dt^2}E(\rho_t)\le 0.
	\end{align}	
\end{theorem}	
The last two theorems are proved in \Cref{concavity}, whereas the proof of \Cref{logsob} is given in \Cref{prooflogsob}.

	\section{Proofs of \Cref{thm:nash} and \Cref{thm:ultra} }\label{Nashproof}
	\begin{proof}[\Cref{thm:nash}] We prove that \Cref{ineq:nash} is satisfied for any Schwartz state $\rho$ with positive Wigner function. For such states $\cL(\rho)$ is well-defined and trace class (see e.g. Propositions 3.14 and 3.15 of \cite{KKW16}), and using the non-commutative Parseval relation, \Cref{NCparseval}, we can write the Dirichlet form (defined in \Cref{Dirichlet})
		as follows:
		\begin{align}
			\mathcal{E}(\rho):=-	\tr(\rho \cL(\rho))=-(2\pi)^{-n}\int_{\cZ} \overline{\cF_{\rho}(z)} \cF_{\cL(\rho)}(z)dz\label{first}
		\end{align}
		Now, using \Cref{commRV}, we find that $\cF_{\cL(\rho)}(z)=-\frac{1}{4}|z|^2\cF_{\rho}(z)$.
		Substituting this into the right hand side of \Cref{first}, we get,
		\begin{align}
			\mathcal{E}(\rho)=\frac{1}{4}(2\pi)^{-n}\int_{\cZ}|z|^2 |\cF_{\rho}(z)|^2dz
		\end{align}
		Now $z\mapsto \cF_{\rho}(z)$ is a Schwartz function, and therefore its inverse Fourier transform is also a Schwartz function. Let us denote by $\cFF_f$ the Fourier transform of an integrable function $f$ on ${\mathbb{R}}^{2n}$:
		\begin{align}
			\cFF_f(z):=(2\pi)^{-n}\int_{\RR^{2n}} f(x) \e^{i(x,z)}dx.
		\end{align}	
		We make use of the following facts. Firstly, the Fourier transform satisfies the following useful identity \begin{align}\label{5.4}
			\cFF_{x\mapsto \partial_{x_j} f(x)}(z)=-iz_j\cFF_f(z),
		\end{align}
		for any integrable, continuously differentiable function $f$, which has an integrable partial derivative $\partial_{x_j} f$, where $z_j$ denotes the $j^{th}$ component of the vector $z \in \cZ$. Secondly, for any square integrable function $h$ on $\RR^{2n}$, we have that
		\begin{align}\label{5.5}
			\|\cFF_h\|_{2}=\|h\|_{2},
		\end{align}
which is the classical Plancherel identity.
		We also employ the well-known polarization identity:
		\begin{align}\label{5.6}
			\int_{\RR^{2n}} \overline{g(x)} h(x) dx = \frac{1}{4}(\|g+h\|^2_{2} - \|g-h\|^2_{2}+i \|g+ih\|^2_{2} -i \|g-ih\|^2_{2}),
		\end{align}
		for any two square integrable functions $g$ and $h$, and the fact that the characteristic function $\cF_\rho$ is equal to 
		$\cFF_{f_\rho}$, where $f_\rho$ denotes the Wigner function of $\rho$ and is a Schwartz function.
		\begin{align} 
			\cE(\rho)&= \frac{1}{4} (2\pi)^{-n}\int_{\cZ} \overline{\cF_\rho(z)}(|z|^2 \cF_\rho(z)) dz\nonumber\\
			&=  -\sum_{i=1}^{2n}\frac{1}{4} (2\pi)^{-n} \int_\cZ \overline{\cFF_{f_{\rho}}(z)} (-i z_j)^2 \cFF_{f_{\rho}}(z) dz\nonumber\\
			&=   -\sum_{i=1}^{2n}\frac{1}{4} (2\pi)^{-n} \int_\cZ \overline{\cFF_{f_{\rho}}(z)}  \cFF_{x\mapsto \partial^2_{x_i^2} f_{\rho}(x)}(z) dz\nonumber\\
			&= -\frac{1}{4} \sum_{i=1}^{2n} (2\pi)^{-n} \int_\cZ {f_{\rho}(x)} \partial_{x_i^2}^2 f_{\rho}(x)dx\nonumber\\
			&=-(2\pi)^{-n} \int_\cZ f_{\rho}(x)\frac{1}{4} \Delta f_{\rho}(x)dx \equiv \frac{1}{4(2\pi)^{n} } \cE_{cl}(f_{\rho}),
		\end{align}
		where $\cE_{cl}(.)$ is the Dirichlet form of the classical heat semigroup and is given by \Cref{DF-cl-heat}. In the above, we have used the notation $\partial_{x_i^2}^2 f(x) = \frac{\partial^2 f(x)}{\partial x_i^2}$; the third line follows from two uses of \Cref{5.4} and the fourth line follows from \Cref{5.5} and \Cref{5.6}.

Moreover, by the non-commutative Parseval relation, \Cref{NCparseval}, and \Cref{5.5}, we have
		\begin{align}\label{5.12}
			\|\rho\|_2^2=\tr(\rho^2)=(2\pi)^{-n}\int_{\cZ}|\cF_{\rho}(z)|^2dz=(2\pi)^{-n}\|\cF_{\rho}\|_{2}^2=(2\pi)^{-n} \|f_{\rho}\|_{2}^2.
		\end{align}
		 Further,
		 \begin{align}
	(2\pi)^n	 \|\rho\|_1=(2\pi)^n=\|f_\rho\|_1,
		 \end{align}
		 which follows from the fact that the Wigner function of a state (as defined through \Cref{wign}) has integral equal to $(2\pi)^n$, and the assumption that $f_\rho$ is positive. Hence,
		\begin{align}\label{nashinvertgauss}
			\|\rho\|_2^{2+2/n}=(2\pi)^{-(n+1)}\|f_{\rho}\|_{2}^{2+2/n}\le (2\pi)^{-n+1}  c_{2n}^{2/n}\cE_{cl}(f_{\rho}) = {8\pi}c_{2n}^{2/n}\cE(\rho),
		\end{align}
		where we used the classical Nash inequality (\ref{eq-cl}), with $n$ replaced by $2n$.
		\qed
	\end{proof}
	
	We now show how \Cref{thm:nash} implies \Cref{thm:ultra}\\\\
	\begin{proof}[\Cref{thm:ultra}]
                 In \cite{KT15} the authors proved ultracontractivity in the finite-dimensional setting starting from a non-commutative version of the Nash inequality. In order to prove an ultracontractivity result for the quantum diffusion semigroup considered in this paper, we follow the ideas of the proof in \cite{KT15}, which in turn closely follows the proof in the classical case studied in \cite{Diaconis1996}. Let $\rho$ be a Schwartz state with a positive Wigner function. By \Cref{theo_KS}, $\rho_t:=\Lambda_t(\rho)$ is also such a state, and \Cref{thm:nash} applies. Let $u(t):=\|\rho_t\|_2^2$; using \Cref{NCparseval} and \Cref{rhorhotcharac}, one can verify that the function $u$ is differentiable and $\dot{u}(t)=-2\eps(\rho_t)$. \Cref{thm:nash} implies that
		\begin{align}
			u^{1+1/n}(t)=\|\rho_t\|^{2+2/n}_2\le  C_{n}  \eps(\rho_t)
			=-\frac{C_{n}}{2} \dot{u}(t) 
		\end{align}
		so that
		\begin{align}
			\frac{d}{dt}\frac{1}{u^{1/n}(t)}\frac{nC_{n}}{2} =-\frac{C_{n}}{2} \frac{\dot{u}(t)}{u(t)^{1+1/n}}\ge 1,
		\end{align}	
		and by integrating both sides of the above inequality from $0$ to $t$, one gets:
		\begin{align}
			\frac{1}{u^{1/n}(t)}\ge \frac{2t}{nC_{n}}+\frac{1}{u^{1/n}(0)}\ge \frac{2t}{nC_{n}},
		\end{align}	
		so that
		\begin{align}
			\|\rho_t\|_2^2\equiv u(t)\le 
\left( \frac{nC_{n}}{2t}\right)^{n}.
		\end{align}	
		The above inequality implies that for any Schwartz state $\rho$ with positive Wigner function,
		\begin{align}\label{at}
			\|\rho_t\|_2\le \left(\frac{nC_{n}}{2t}\right)^{n/2}\equiv \kappa_n t^{-n/2}.
		\end{align}
		In particular, this implies that
	\begin{align}
	\|\rho_t\|_{\infty}\le \|\rho_t\|_2\le \kappa_n t^{-n/2}.
	\end{align}
Therefore,
\begin{align}
S(\rho_t)=-\tr(\rho_t\ln \rho_t)\ge -\ln\|\rho_t\|_\infty\ge \frac{n}{2}\ln\left(\kappa_n^{-2/n}t\right).
\end{align}
	\qed
\end{proof}

\section{Proofs of \Cref{blachmanstam} and \Cref{concavqep}}\label{concavity}
Our main ingredient in the proof of \Cref{logsob} is the concavity of the entropy power stated in \Cref{concavqep}. This concavity has a classical analogue, first proved by Costa in \cite{Costa1985}. Later, Dembo \cite{Dembo1991,Dembo1989} simplified the
proof, by an argument based on the so-called Blachman-Stam inequality \cite{Blachman1965}. More recently, Villani \cite{Villani2000} gave a direct proof of the same inequality. Our proof of \Cref{concavqep} can be interpreted as a quantum version of the proof in \cite{Dembo1989}, and the intermediate result \Cref{blachmanstam} can be seen as a quantum analogue of the Blachman-Stam inequality in the case of added Gaussian noise. 
  	\begin{lemma}\label{difftwice}
  		Let $\rho$ be an invertible Gaussian state. Then for any $j \in \{1,2, \ldots, 2n\}$, the function $\theta\mapsto D(\rho||\rho_{R_j}^{(\theta)})$ is twice differentiable at $0$, where $\rho_{R_j}^{(\theta)}$ is defined through \Cref{shift-states}. Moreover, the operator $\left(\rho[R_j,[R_j,\ln\rho]]\right)$ is trace class, and the divergence-based quantum Fisher information defined through \Cref{Fisherinfo} is given by
  		\begin{align}
  		\cJ(\{\rho_{R_j}^{(\theta)}\})=\tr\left(\rho[R_j,[R_j,\ln\rho]]\right).
  	\end{align}	
In addition, $\left.\frac{d}{d\theta}D(\rho||\rho_{R_j}^{(\theta)})\right|_{\theta=0}=0$.
  	\end{lemma}
The proof of the above lemma is given in \Cref{app:lem4}. It follows from it that for any Gaussian state $\rho$ satisfying the condition of \Cref{difftwice} the entropy variation rate, defined by \Cref{entropy-var-rate} is finite, and is given by the following expression
 \begin{align}\label{JLL}
 J(\rho)=-4\tr(\LL (\rho)\ln\rho),
 \end{align}
 \medskip

\noindent
 We are now ready to prove \Cref{blachmanstam}. The proof we give can be seen as a quantum analogue of the proof given by Stam in \cite{Stam1959}.\\

\begin{proof}[\Cref{blachmanstam}]
In the following, we define a vector $z=(q_1,p_1,\ldots, q_n,p_n) \in \cZ$ simply as $(q,p)$, and for any $a \in {\mathbb{R}}$ we define
\begin{align}
(q,p;q_j-a)&:=(q_1,p_1,\ldots,q_{j-1},p_{j-1}, q_j-a, p_j,q_{j+1}p_{j+1} \ldots,q_n,p_n),\\
(q,p;p_j-a)&:=(q_1,p_1,\ldots,q_{j-1},p_{j-1}, q_j, p_j-a,q_{j+1}p_{j+1} \ldots,q_n,p_n).
\end{align}
Then for any $a \in {\mathbb{R}}$ we define the following functions:
\begin{align}
g^{a}_{P_j,t/2}: (q,p)\mapsto g_{t/2}(q,p;q_j-a); \quad g^{a}_{Q_j,t/2}: (q,p)\mapsto g_{t/2}(q,p;p_j+ a)
\end{align}
Further, accordingly denoting the Weyl operator $V(z)$ as $V(q,p)$, we have for any $\theta \in {\mathbb{R}}$ and $\alpha, \beta >0$,
\begin{align}
\rho_{P_j}^{(\theta \alpha)}\ast g^{(\theta \beta)}_{P_j,t/2} &=\int_{{\mathbb{R}}^n}\int_{{\mathbb{R}}^n} dqdp g_{t/2}(q,p;q_j-\theta\beta) V{(q,p)}\rho_{P_j}^{(\theta \alpha)}V(-q,-p)\nonumber\\
 &=\int_{{\mathbb{R}}^n}\int_{{\mathbb{R}}^n} dq\,dp\, g_{t/2}(q,p)\cA(q,p;q_j+\theta\beta)(\rho_{P_j}^{(\theta \alpha)})
 \end{align}
 where we made a change of variable $ q_j-\theta\beta \rightarrow q_j$ and denoted the conjugation $V(z)(\cdot)V(-z)$ by the automorphism $\cA(z)(\cdot)$ in the last line. 
Then using the Weyl-Segal CCR (\Cref{CCRrelation}), we obtain
 \begin{align}
V{(q,p;q_j+\theta\beta)}\e^{i\theta\alpha {P_j}}
 &= e^{ip_j \theta \alpha/2}V\left(q,p;q_j+\theta(\alpha + \beta)\right)
\end{align}
Hence,
\begin{align}
\rho_{P_j}^{(\theta \alpha)}\ast g^{(\theta \beta)}_{P_j,t/2}&=\int_{{\mathbb{R}}^n}\int_{{\mathbb{R}}^n}dqdp g_{t/2}(q,p)\cA(q,p;q_{j}+\theta(\alpha + \beta))(\rho)\nonumber\\
&=e^{i\theta(\alpha + \beta)P_j} \left(\int_{{\mathbb{R}}^n}\int_{{\mathbb{R}}^n}dq dp\, g_{t/2}(q,p) \cA(q,p)(\rho)\right)e^{-i\theta(\alpha + \beta)P_j} \nonumber\\
&=(\rho_t)_{P_j}^{\left(\theta(\alpha + \beta)\right)}.
\label{astast}
\end{align}
This implies that 
\begin{align}
	D\left(\rho_t||\rho_{P_j}^{(\theta\alpha)}\ast g_{P_j,t/2}^{(\theta\beta)}\right)= D\left(\rho_t||(\rho_t)_{P_j}^{\left(\theta(\alpha + \beta)\right)}\right).
\end{align}
Similarly, we can show that
\begin{align}
	D\left(\rho_t||\rho_{Q_j}^{(\theta\alpha)}\ast g_{Q_j,t/2}^{(\theta\beta)}\right)= D\left(\rho_t||(\rho_t)_{Q_j}^{\left(\theta(\alpha+ \beta)\right)}\right).
\end{align}

Further, using the data-processing inequality for the relative entropy, and its additivity under tensor product one can show that
\begin{align}\label{6.11}
	D\left(\rho_t||\rho_{R_j}^{(\theta\alpha)}\ast g_{R_j,t/2}^{(\theta\beta)}\right)
&\leq D\left(\rho||\rho_{R_j}^{(\theta\alpha)}\right) + D_{\rm{KL}}\left(g_{t/2}||g_{R_j,t/2}^{(\theta\beta)}\right),
\end{align}
where $D_{\rm{KL}}(f||g)$, for pdfs $f$ and $g$, denotes the Kullback-Leibler divergence:
\begin{align}
		D_{\rm{KL}}(f||g):=\int f(x)\ln\frac{f(x)}{g(x)}dx.
	\end{align}
A rigorous proof of the inequality (\ref{6.11}), using properties of the relative entropy of normal states of general von Neumann algebras, is given in 
\Cref{proofineqpos}. From (\ref{6.11}) it can be shown that 
		\begin{align}	\label{theinequality}
			(\alpha+\beta)^2J(\rho\ast g_{t/2})\le \alpha^2J(\rho)+\beta^2 J_{{cl}}(g_{t/2}),
		\end{align}
where for any positive differentiable probability density function $f$ 
\begin{align}\label{cl-FI}
	J_{{cl}}(f):=\int_{\RR^{2n}}\frac{|\nabla f|^2}{f}dz
	\end{align}
	denotes its (classical) Fisher information. The inequality (\ref{theinequality}) is also derived in \Cref{proofineqpos}. A straightforward computation of $J_{{cl}}(g_{t/2})$ then yields the result.
\qed
\end{proof}
\medskip

\noindent
We now give the proof of \Cref{concavqep}.	

\begin{proof}[\Cref{concavqep}]
		By \Cref{qbruijn}, it is suffices to prove that $t \mapsto J(\rho_t)$ is differentiable on $[0, \infty)$ and for any $t\ge0$,
		\begin{align}\label{equivconcav}
			\frac{d}{dt}J(\rho_t)+\frac{1}{4n}J(\rho_t)^2\le 0
		\end{align}	
Let $\varepsilon>0$. By setting $\alpha=\frac{1}{J(\rho_t)}$ and $\beta=\frac{\varepsilon}{4n}$ in (\ref{qfisherinfoin}), we get
\begin{align}
	\left( \frac{1}{J(\rho_t)}+\frac{\varepsilon}{4n}\right)^2 J(\rho_{t+\varepsilon})\le \frac{1}{J(\rho_t)}+\frac{\varepsilon}{4n}
\end{align}	
which implies
\begin{align}\label{Jeps}
\frac{J(\rho_{t+\varepsilon})-J(\rho_t)}{\varepsilon}\le \frac{-J(\rho_t)J(\rho_{t+\varepsilon})}{4n}.
\end{align}

For an invertible centered Gaussian state, we know from \Cref{JLL} and \Cref{theoremgaussstate} that 
\begin{align}
	J(\rho_t)&=-\tr(\LL(\rho_t)\ln \rho_t)\nonumber \\
	&=-\frac{1}{4}\tr\left ( \int_{\RR^{2n}} g_{t/2}(z) V(z) \rho V(-z) \sum_{ijk} (\Gamma_t)_{ij}[ R_k,[R_k,R_i R_j]] dz\right)\nonumber\\
	&=-\frac{1}{4} \sum_{ijk}(\Gamma_t)_{ij} \int_{\RR^{2n}} g_{t/2}(z) \tr(V(z) \rho V(-z)   [R_k,[R_k,R_i R_j]]) dz,
	\end{align}
	where we use Fubini's theorem to interchange the trace and the integral in the last line. Here $\Gamma_t$ is the matrix associated to the Gaussian state $\rho_t$ defined similarly to \Cref{theoremgaussstate}. Now for any $t_0>0$, $t\mapsto \Gamma_t$ and $t\mapsto g_{t/2}(z)$, $\forall$ $z \in {\mathbb{R}}^{2n}$, are differentiable on $[t_0,\infty)$ and $|\partial_t g_{t/2}(z)|\le |\partial_t g_{t/2}(z)|_{t=t_0}|$ for all $t\ge t_0$, with
	\begin{align}
		\forall i,j\in \{ 1, \ldots,2n\},\quad \sum_k \int_{\RR^{2n}} |\partial_{t}g_{t/2}(z)\tr (V(z) \rho V(-z)[R_k,[R_k,R_i R_j]])|dz|_{t=t_0}<\infty.
		\end{align}
	Hence, $t\mapsto J(\rho_t)$ is differentiable on $(0,+\infty)$. Using the above, the result follows by taking the limit $\varepsilon\to 0^+$ in \Cref{Jeps}.
\qed
	\end{proof}	
	\section{Proof of \Cref{logsob} and \Cref{iso}}\label{prooflogsob}
	The proof of \Cref{logsob} relies on the concavity of the quantum entropy power with respect to time, and is analogous to Toscani's proof of the classical log-Sobolev inequality in the form given by \Cref{Tos1}.
\smallskip

\noindent	
\begin{proof}[\Cref{logsob} and \Cref{iso}]
		We know by de Bruijn's identity (\Cref{quantumdebruijin} ) that for any invertible, centered Gaussian state $\rho$ evolving under the action of the quantum diffusion semigroup,
		\begin{align}
			\frac{d}{ds} E(\rho_s)=\frac{1}{4n}J(\rho_s)E(\rho_s).
			\end{align}
			Moreover by the concavity of the entropy power (\Cref{concavqep}), 
			\begin{align}
				\left.\frac{d}{ds} E(\rho_s)\right|_{s=0} \ge \frac{E(\rho_t)-E(\rho)}{t},\qquad \forall \, t>0.
				\end{align}
				However, by Corollary III.4 of \cite{Koenig2012}, $E(\rho_t)=\e t/2+\mathcal{O}(1)$. Combining these facts, we get the following inequality:
		\begin{align}
			J(\rho)\ge \e^{-\frac{1}{n} S(\rho)}2\e n,
		\end{align}
		which is the statement of \Cref{iso}. Now use the inequality $\e^{-x}\ge 1-x$ for every $x$ to deduce that
		\begin{align}
			\tr(\rho\ln \rho)+n\le \frac{ J(\rho)}{2\e}\label{fisherLSI}
		\end{align}
		
		\qed
	\end{proof}
	This equation is analogous to the classical Fisher-information log-Sobolev inequality (see Equation 24 in \cite{Toscani}), which was proven to be equivalent to the Gaussian log-Sobolev inequality for the Ornstein-Uhlenbeck Markov semigroup (\cite{Chafai2005}, see also \cite{Weissler1978}). 

\section*{Remark}
\Cref{iso,blachmanstam,concavqep} were independently and concurrently obtained by Huber, K\"{o}nig and Vershynina in \cite{koenig2016}.	
\section*{Acknowledgments}
ND is grateful to the organizers of the {\em{48 Symposium on Mathematical Physics; Gorini-Kossakowski-Lindblad-Sudarshan Master Equation - 40 Years After}}, held in Torun in June 10-12, 2016, where the results of this paper were first presented. She would also like to thank Robert Alicki, Dariusz Chruściński, Franco Fagnola, Bassano Vacchini and Reinhard Werner for helpful discussions during the Symposium. CR would like to thank St\'{e}phane Boucheron and Djalil Chafa\"{i} for very useful exchanges. YP would like to thank Ivan Bardet for helpful discussions, and acknowledges the support of ANR contract ANR-14-CE25-0003-0. The authors are grateful to Eric Carlen and Michael Loss for pointing out an inaccurate remark in a previous version.

\appendix
\section{Proof of \Cref{difftwice}}\label{app:lem4}
\begin{proof}
	Recall that any Gaussian state (not necessarily centered) can be obtained from a centered Gaussian state via a similarity transformation involving the Weyl operators (cf.~\Cref{gausstogausscentered}).
	Hence, by invariance of the relative entropy under unitary operations, functional calculus, \Cref{commRV} and \Cref{CCRrelation}, we can reduce the proof to that of the case of centered Gaussian states. Suppose then that $\rho$ is a centered, invertible Gaussian state. Hence, by \Cref{theoremgaussstate}:
	\begin{align}
		\tr(\rho|\ln\rho_{R_j}^{(\theta)}|)&=\tr\left( \rho| \e^{iR_j\theta} (\ln(C)-R^T\Gamma R) \e^{-iR_j\theta} |\right)\nonumber \\
		&	\le |\ln C| +\tr (\rho|  R^T\Gamma R|)\nonumber\\
		&\le |\ln C| +\sum_{k,l=1}^{2n}|\Gamma_{kl}| \tr  (\rho|R_kR_l|)\nonumber\\
		&\le |\ln C| +\sum_{k,l=1}^{2n}|\Gamma_{kl}|\tr \rho(R_lR_k^2R_l)<\infty.
	\end{align}
	The last line follows from the Cauchy-Schwarz inequality and the fact that a Gaussian state has finite moments of each order, as its characteristic function is smooth, so that
	\begin{align}
		\forall k,l\in \llbracket 1,2n\rrbracket,\quad	\tr(\rho^{1/2}\rho^{1/2}|R_kR_l|)\le \tr(\rho)\tr(|R_kR_l|^2\rho)=\tr(R_lR_k^2R_l\rho)<\infty.
	\end{align}
	Hence, the quantity $\tr(\rho\,\ln\rho_{R_j}^{(\theta)})$ is well-defined for any $\theta$. Now, for any $\theta,\varepsilon>0$,
	\begin{align}
		\frac{1}{\varepsilon}\left[ D(\rho||\rho_{R_j}^{(\theta+\varepsilon)})\label{SSterm} -D(\rho||\rho_{R_j}^{(\theta)})\right]&=\frac{1}{\varepsilon}[\tr(\rho\,\ln\rho_{R_j}^{(\theta)})-\tr(\rho\,\ln \rho_{R_j}^{(\theta+\varepsilon)})]\nonumber \\
		&=\frac{1}{\varepsilon}\left\{\tr\rho\,\e^{i\theta R_j}(I-\e^{i\varepsilon R_j})\ln\rho\,\e^{-i(\theta+\varepsilon)R_j}\right.\nonumber\\
		&\left.+\tr\rho\,\e^{i\theta R_j}\ln\rho \e^{-i\theta R_j}(I-\e^{-i\varepsilon R_j})\right\}   \\
		&= \frac{1}{\varepsilon} \sum_{kl} \Gamma_{kl} \left\{\tr R_kR_l \e^{-i(\theta+\varepsilon)R_j} \rho \e^{i\theta R_j} (\e^{i\varepsilon R_j}-I)\right.\label{firstterm}\\
		&\left.\qquad+\tr\e^{-i\theta R_j}\rho \e^{i\theta R_j} R_kR_l  (\e^{-i\varepsilon R_j}-I)\right\}\label{secondterm},
	\end{align}
	where we used \Cref{exprrho} in the last line above. Now, $\e^{-i\theta R_j}\rho\,\e^{i\theta R_j}$ is also a Gaussian state (cf.~ \Cref{gausstogausscentered}), and hence the operator $A:=\e^{-i\theta R_j}\rho\,\e^{i\theta R_j} R_k R_l$ is trace-class for each $k,l\in\llbracket 1,...,2n\rrbracket$. Moreover, as $\e^{-i\theta R_j}\rho\,\e^{i\theta R_j}$ is Gaussian, $\tr(|A| R_j^2)<\infty$ (use e.g. polar decomposition of $A$). Hence, up to a decomposition of $A$ into its positive and negative parts, and normalizing each of these parts, one finds, by use of \Cref{lemma5.4.2} and the Cauchy-Schwarz inequality, that
	\begin{align}
		\frac{1}{\varepsilon}	\tr\e^{-i\theta R_j}\rho \e^{i\theta R_j} R_kR_l  (\e^{-i\varepsilon R_j}-I)\underset{\varepsilon\to 0}{\rightarrow} -i\tr \e^{-i\theta R_j}\rho \e^{i\theta R_j} R_kR_l R_j\label{secondtterm}
	\end{align}
	Let us now focus on \Cref{firstterm}. We can merge the exponentials depending on $\varepsilon$ in the following way: first we observe by \Cref{Weyl} that $\e^{-i\varepsilon R_j}$ is equal to $V(z^\varepsilon_j)$, where $z^\varepsilon_j$ is the $2n$ dimensional vector with entries
	
	\begin{align}
		(z_j^\varepsilon)_r:=\left\{ \begin{aligned}
			& -\varepsilon\quad \text{ if  }\,r=j-1,\quad j\text{ even},\\
			&+\varepsilon\quad \text{ if }\,r=j+1,\quad j\text{ odd},\\
			&0\quad \text{ otherwise.}
		\end{aligned}\right.
	\end{align}

	Suppose for example that $j$ is even. Then
	\begin{align}
		R_kR_l \e^{-i\varepsilon R_j}&=R_kR_l V(z^\varepsilon_j)\nonumber \\
		&= V(z_j^\varepsilon)V(z_j^{-\varepsilon}) R_k V(z_j^\varepsilon)V(z_j^{-\varepsilon})R_l V(z_j^\varepsilon)\label{shift}\\
		&= V(z_j^\varepsilon) (\varepsilon \delta_{j-1,k}I+R_k)(\varepsilon\delta_{j-1,l}I +R_l),\nonumber
	\end{align}
	where we used \Cref{commRV} in the last line. Hence, substituting  this into \Cref{firstterm},
	
	\begin{align}
		\tr(R_kR_l \e^{-i(\theta+\varepsilon)R_j}\rho\,\e^{i\theta R_j}(\e^{i\varepsilon R_j}-I))=&	\delta_{j-1,k}\delta_{j-1,l} \varepsilon^2 \tr( \e^{-i\theta R_j}\rho \e^{i\theta R_j}(I-\e^{-i\varepsilon R_j}))\nonumber \\
		&+\varepsilon\delta_{j-1,l} \tr(R_k \e^{-i\theta R_j}\rho \e^{i\theta R_j}(I-\e^{-i\varepsilon R_j}))\\
		&+\varepsilon\delta_{j-1,k} \tr(R_l \e^{-i\theta R_j}\rho \e^{i\theta R_j}(I-\e^{-i\varepsilon R_j}))\nonumber \\
		&+\tr(R_kR_l  \e^{-i\theta R_j}\rho \e^{i\theta R_j}(I-\e^{-i\varepsilon R_j}))\nonumber
	\end{align}
	Then, by arguments very similar to the ones used to prove \Cref{secondtterm}, we find that 
	\begin{align}
		\frac{1}{\varepsilon} \tr \left(R_kR_l \e^{-i(\theta+\varepsilon)R_j} \rho \e^{i\theta R_j} (\e^{i\varepsilon R_j}-I)\right)\underset{\varepsilon\to 0}{\rightarrow}  i \tr(R_k R_l \e^{-i\theta R_j}\rho\,\e^{i\theta R_j}R_j).\label{firsttterm}
	\end{align}
	One easily shows that the result is the same for $j$ odd. Finally using \Cref{secondtterm} and \Cref{firsttterm} into \Cref{SSterm}, we obtain
	\begin{align}
		\frac{1}{\varepsilon} [D(\rho||\rho_{R_j}^{(\theta+\varepsilon)})-D(\rho||\rho_{R_j}^{(\theta)})]\underset{\varepsilon\to 0}{\rightarrow}-i \tr(\e^{-i\theta R_j}\rho\,\e^{i\theta R_j}[R_j,\ln\rho])
	\end{align}
	By similar arguments, we can show that:
	\begin{align}
		\frac{d^2}{d\theta^2}D(\rho||\rho_{R_j}^{(\theta)})|_{\theta=0}=	\lim_{\theta\to 0}\frac{-i}{\theta}\tr((\e^{-i\theta R_j}\rho\,\e^{i\theta R_j}-\rho)[R_j,\ln\rho])= \tr(\rho[R_j,[R_j,\ln\rho]]).
	\end{align}	
	Indeed, 
	\begin{align}
		\left.\frac{d^2}{d\theta^2}D(\rho||\rho_{R_j}^{(\theta)})\right|_{\theta=0}&=\lim_{\theta\to0}\frac{1}{\theta}  \left.\frac{d}{d\varepsilon}\right|_{\varepsilon=0}D(\rho||\rho_{R_j}^{(\theta+\varepsilon)})-\left.\frac{d}{d\varepsilon}\right|_{\varepsilon=0}D(\rho||\rho_{R_j}^{(\varepsilon)})\nonumber\\
		&=\lim_{\theta\to 0}\frac{-i}{\theta} \tr((\e^{-i\theta R_j}\rho\,\e^{i\theta R_j}-\rho)[R_j,\ln\rho]).
	\end{align}
	Now, for any $\theta>0$, 
	\begin{align}
		\e^{-i\theta R_j}\rho\,\e^{i\theta R_j}-\rho=(\e^{-i\theta R_j}-I)\rho \e^{i\theta R_j}+\rho(\e^{i\theta R_j}-I)
	\end{align}
	So that
	\begin{align}
		(\e^{-i\theta R_j}\rho\,\e^{i\theta R_j}-\rho)[R_j,\ln\rho]=(\e^{-i\theta R_j}-I)\rho \e^{i\theta R_j} [R_j,\ln\rho]+\rho(\e^{i\theta R_j}-I) [R_j,\ln \rho].\label{secondderive}
	\end{align}
	Let us focus on the second term of the right hand side first:
	\begin{align}
		[R_j,\ln \rho]\rho=-\sum_{k,l=1}^{2n}\Gamma_{kl}[R_j,R_kR_l]\rho
	\end{align}
	is trace class, and $\tr(|[R_j,\ln \rho]\rho|R_j^2)<\infty$, since $\rho$ is Gaussian. Indeed, by polar decomposition:
	\begin{align} 
		\tr(|[R_j,\ln \rho]\rho|R_j^2)&\le \sum_{k,l=1}^{2n} |\Gamma_{k,l}| \tr(|[R_j,R_kR_l]\rho|R_j^2)\nonumber\\
&=\sum_{k,l=1}^{2n} |\Gamma_{k,l}| \tr([R_j,R_kR_l]\rho UR_j^2) <\infty,
	\end{align}
which follows from use of the Cauchy-Schwarz inequality.
	Hence, by 
	\Cref{lemma5.4.2} and the Cauchy-Schwarz inequality,
	\begin{align}
		-	\frac{i}{\theta}\tr(\rho(\e^{i\theta R_j}-I)[R_j,\ln \rho])\underset{\theta\to 0}{\to}\tr(\rho R_j[R_j,\ln\rho]).\label{firsttermlimit}
	\end{align}
	Let us now focus on the first term of the right hand side of \Cref{secondderive}:
	\begin{align}
		\e^{i\theta R_j}[R_j,\ln\rho]&=-\sum_{k,l=1}^{2n} \e^{i\theta R_j} \Gamma_{kl}[R_j,R_kR_l\label{firsttermeq}]\nonumber\\
		&=-\sum_{k,l=1}^{2n} \e^{i\theta R_j} \Gamma_{kl}(R_jR_kR_l-R_kR_lR_j)\nonumber\\
		&=-\sum_{k,l=1}^{2n} \Gamma_{kl} V(z_j^{-\theta})(R_jR_kR_l-R_kR_lR_j).
	\end{align}
	However, for $j$ even, recall that
	\begin{align}
		V(z_j^{-\theta})R_kR_l&=V(z_j^{-\theta})R_k V(z_j^{\theta})
		V(z_j^{-\theta})R_l V(z_j^{\theta})V(z_j^{-\theta})\nonumber\\
		&=(\theta \delta_{j-1,k+R_k})(\theta \delta_{j-1,l}+R_l)V(z_j^{-\theta}),
	\end{align}
	where we used \Cref{shift} in the last line. Replacing in \Cref{firsttermeq}, we get 
	\begin{align}
		\e^{i\theta R_j}[R_j,\ln\rho]=&-\sum_{k,l=1}^{2n} \Gamma_{kl}R_j(\theta \delta_{j-1,k}I +R_k)(\theta \delta_{j-1,l}I +R_l)V(z_j^{-\theta})\nonumber\\
		&+\sum_{k,l=1}^{2n} \Gamma_{kl}(\theta \delta_{j-1,k}I +R_k)(\theta \delta_{j-1,l}I +R_l)R_jV(z_j^{-\theta})
	\end{align}
	Replacing in the first term of the right hand side of \Cref{secondderive}, we then get, by an argument similar to the one leading to \Cref{firsttterm}, that:
	\begin{align}
		\frac{-i}{\theta}\tr((\e^{-i\theta R_j}-I)\rho\,\e^{i\theta R_j} [R_j,\ln \rho] )\to_{\theta\to 0}=-\tr(\rho[R_j,\ln\rho]R_j).\label{secondtermlim}
	\end{align}
	The same calculation can be carried without difficulty for $j$ odd. Regrouping \Cref{firsttermlimit} and \Cref{secondtermlim}, we finally obtain
	\begin{align}
		\left. \frac{d^2}{d\theta^2}\right|_{\theta=0}D(\rho||\rho_{R_j}^{\theta})=\tr(\rho[R_j,[R_j,\ln\rho]])
	\end{align} 
		\qed
\end{proof}

\section{Araki's generalized relative entropy}\label{relative}
In this section, we review the relative entropy introduced by Araki in \cite{Araki1976} in the general context of normal states on von Neumann algebras. This relative entropy reduces to the classical Kullback-Leibler divergence in the case of equivalent probability measures, and to Umegaki's quantum relative entropy \cite{Umegaki1962} in the case of states defined through density operators $\rho,\sigma$ defined on a fixed separable Hilbert space, where $\supp\rho\subset \supp\sigma$. We heavily used the Section 5 of \cite{Jaksicb} to write up this section, but also invite the interested reader to have a look at \cite{Takesaki,Aspects2003} for further details.
\subsection{Basic definitions and modular structure}
We start by recalling that a (concrete) \textit{von Neumann algebra} $\mathfrak{V}$ on a Hilbert $\mathcal{H}$ is a  self-adjoint, weakly closed subalgebra of the algebra $\cB(\cH)$ of bounded operators on $\cH$, which contains the identity operator $I$. Let $\mathfrak{V}$ and $\mathfrak{W}$ be two von Neumann algebras, then a unital map $\Phi:\mathfrak{V}\rightarrow \mathfrak{W} $ is called a \textit{Schwarz map} if for any $A\in\mathfrak{V}$,
\begin{align}
	\Phi(A^*A)\ge \Phi(A)^*\Phi(A)
\end{align}	
A functional $\omega$ on a von Neumann algebra $\mathfrak{V}$ is said to be {\em{normal}} if for any increasing net $\{A_\alpha\}$ of positive operators in $\mathfrak{V}$, $\omega(\operatorname{l.u.b.} A_\alpha)=\operatorname{l.u.b.} \omega(A_\alpha)$, where $\operatorname{l.u.b}.$ stands for the least upper bound of a net. Normal functionals form a subspace of the space of functionals $\mathfrak{V}^*$ on $\mathfrak{V}$, called the \textit{predual} of $\mathfrak{V}$, and denoted by $\mathfrak{V}_*$. The subset of positive normal functionals is denoted by $\mathfrak{V}^+_*$.  A \textit{state} on a von Neumann algebra $\mathfrak{V}$ is a positive linear functional $\omega:\mathfrak{V}\to \CC$ such that $\omega(I)=1$, where $I$ is the identity operator on $\mathfrak{V}$. A normal state $\omega$ is characterized by the existence of a density operator $\rho$, i.e. a non-negative trace-class operator $\rho$ on $\mathcal{H}$ with $\tr(\rho)=1$, such that
\begin{align}\label{statenormaldef}
	\omega(A)=\tr(\rho A),\quad A\in\mathfrak{V}
\end{align}
A normal state is said to be \textit{faithful} if for any positive element $A\in\mathfrak{V}$, $\omega(A)=0$ implies that $A=0$. A von Neumann algebra \textit{in standard form} is a quadruple $(\mathfrak{V},\cH, J, \cH^+)$ where $\cH$ is a Hilbert space, $\mathfrak{V}\subset \cB(\cH)$ is a von Neumann algebra, J is an anti-unitary involution on $\cH$ and $\cH^+$ is a cone in $\cH$ such that:\\\\
(i) $\cH^+$ is self-dual, i.e. $\cH^+=\{x\in\cH|\langle y,x\rangle\ge 0~ \forall y\in\cH^+\}$,\\
(ii) $J\mathfrak{V}J=\mathfrak{V}'$,\\ 
(iii) $JAJ=A^*$ for $A\in\mathfrak{V}\cap\mathfrak{V}'$.\\
(iv) $Jx=x$ for $x\in\cH^+$,\\
(v) $AJA\cH^+\subset \cH^+$ for $A\in\mathfrak{V}$,\\\\
where in (ii), $\mathfrak{V}'$ denotes the commutant of $\mathfrak{V}$ in $\cB(\cH)$.
A quadruple $(\pi,\cH,J,\cH^+)$ is a \textit{standard representation} of the von Neumann algebra $\mathfrak{V}$ if $\pi:\mathfrak{V}\to \cB(\cH)$ is a faithful representation and $(\pi(\mathfrak{V}),\cH,J,\cH^+)$ is in standard form. It is a celebrated result in operator algebras that a standard representation always exists, which means that any von Neumann algebra can be seen as a standard von Neumann algebra acting on a Hilbert space $\cH$, up to some isomorphism $\pi$. Moreover, if $(\pi_1,\cH_1,J_1,\cH^+_1)$ and $(\pi_2,\cH_2,J_2,\cH^+_2)$ are two standard representations of $\mathfrak{V}$, then there exists a unique unitary operator $U:\cH_1\to \cH_2$ such that $U\pi_1(A)U^*=\pi_2(A)$ for all $A\in\mathfrak{V}$, $UJ_1U^*=J_2$, and $U\cH_1^+=\cH_2^+$. The following basic theorem of operator algebra is useful to define the relative entropy.
\begin{theorem}\label{Omega2}
	Let $\mathfrak{V}$ be a von Neumann algebra in standard form. For any positive normal functional $\omega$ on $\mathfrak{V}$, there exists a unique $\Omega_\omega\in\cH^+$ such that 
	\begin{align}
		\omega(A)=\langle \Omega_{\omega}, A\Omega_{\omega}\rangle
	\end{align}
	for all $A\in \mathfrak{V}$. The map $\mathfrak{V}_*^+\ni \omega\mapsto \Omega_{\omega}\in\cH^+$ is a bijection and
	\begin{align}
		\|\Omega_\omega-\Omega_\nu\|^2\le \|\omega-\nu\|\le \|\Omega_\omega-\Omega_\nu\|\|\Omega_\omega +\Omega_\nu\|,
	\end{align}
	where 
	\begin{align}
		\|\omega\|:=\sup_{A\in \mathfrak{V},\|A\|=1}|\omega(A)|.
	\end{align}
\end{theorem}	
We now recall the definition of Araki's relative modular operator \cite{Araki1976}. For $\nu,\omega\in \mathfrak{V}_*^+$, define $S_{\nu|\omega}$ on the domain $\mathfrak{V}\Omega_\omega+(\mathfrak{V}\Omega_\omega)^\perp$ by
\begin{align}\label{SS}
	S_{\nu|\omega}(A\Omega_\omega+\Theta)=P_\omega A^*\Omega_\nu,
\end{align}
for all $A\in \mathfrak{V}$, $\Theta\in(\mathfrak{V}\Omega_\omega)^\perp$, and where $P_\omega$ is the support of $\omega$, namely the projection defined by
\begin{align}
	P_\omega:=\inf \{P\in\mathfrak{V}|P\text{ is a projection and }\omega(I-P)=0\}
\end{align}
$S_{\nu|\omega}$ is a densely defined anti-linear operator. It is closable and we denote its closure by the same symbol. The positive operator
\begin{align}\label{Delta}
	\Delta_{\nu|\omega}:=S_{\nu|\omega}^*S_{\nu|\omega}
\end{align}
is called the \textit{relative modular operator}. We denote $\Delta_\omega:=\Delta_{\omega|\omega}$. 
We give two examples of states which are going to be useful in our proof:
\begin{example}[Classical probability theory]
	Let $(\Omega, \cF, \mu)$ be a $\sigma$-finite measure space. Define the von Neumann algebra $L^{\infty}(\Omega, \cF,\mu)$ of bounded, measurable functions acting by pointwise left multiplication on the Hilbert space $L^2(\Omega,\cF,\mu)$ of square integrable functions. This means that to any $f\in L^{\infty}(\Omega, \cF,\mu)$, one can associate the operator $M_f: L^2(\Omega,\cF,\mu)\to L^2(\Omega,\cF,\mu)$ defined by
	\begin{align}
		M_f(g)(x)=f(x)g(x),\quad \forall g\in L^2(\Omega,\cF,\mu), \quad x\in\Omega.
	\end{align}
	The map $M:L^\infty(\Omega,\cF,\mu)\ni f\mapsto M_f$ is a faithful representation, so that
	the quadruple $(M(L^\infty(\Omega,\cF,\mu)),L^2(\Omega,\cF,\mu),\bar{\hphantom{n} },L^2(\Omega,\cF,\mu)_+)$ is in standard form, where $\bar{\hphantom{n} }$ is the usual complex conjugation and $L^2(\Omega,\cF,\mu)_+$ denotes the set of positive square-integrable functions with respect to $\mu$. We use the same notation for normal states and their associated probability measures, which is justified by the above mentioned isometry. Now any function $f\in L^1(\Omega, \cF,\mu)$ such that $\|f\|_{1}=1$ represents a state $\omega_f$ on $L^\infty(\Omega,\cF,\mu)$ via the relation:
	\begin{align}\label{statesclass}
		\omega_f(h):=\int_{\Omega} f(x)h(x)\mu(dx), \quad \forall h\in L^\infty(\Omega,\cF,\mu)
	\end{align}	
	Indeed, one has $\int_\omega f(x) 1 dx=\|f\|_{1}=1$, $1$ encoding for the identity operator in $L^\infty(\Omega, \cF,\mu)$, and one easily checks that
	\begin{align}
		\omega_f(h)=\langle \sqrt{f}, M_h\sqrt{f}\rangle=:\tr(|\sqrt{f}\rangle\langle \sqrt{f}| M_h),
	\end{align}	
	So that $\omega_f$ is indeed normal. Actually any normal functional can be written in this way, so that $L^\infty(\Omega,\cF,\mu)_*\cong L^1(\Omega,\cF,\mu)$. This is the space of (complex) measures absolutely continuous with respect to $\mu$. Then for any measure $\nu \in L_+^\infty(\Omega,\cF,\mu)_*$, one associates a positive integrable function $f$ which is the Radon-Nikodym derivative of $\nu$ with respect to $\mu$:
	\begin{align}
		f:=\frac{d\nu}{d\mu}.
	\end{align}
	One easily checks then that the vector $\Omega_\nu$ defined in \Cref{Omega2} reduces to
	\begin{align}
		\Omega_\nu=\left(\frac{d\nu}{d\mu}\right)^{1/2}.
	\end{align}
	Using \Cref{SS} and \Cref{Delta}, one then finds that the relative modular operator of $\nu$ with respect to $\mu$ is defined by:
	\begin{align}
		\Delta_{\nu|\mu}(f):=\frac{d\nu}{d\mu}f,\qquad f\in L^2(\Omega,\cF,\mu).
	\end{align}
\end{example}
\begin{example}[Quantum systems]
	Take $\mathfrak{V}:=\cB(\mathcal{H})$ to be the von Neumann algebra of all bounded operators on a separable Hilbert space $\cH$. Then any density operator $\rho$ provides a state via \cref{statenormaldef}. Then $\mathfrak{V}_*$ is identified with the space of trace-class operators acting on $\cH$. The map $\pi: \cB(\cH)\ni A\mapsto L_A$, where for any $A\in\cB(\cH)$ $L_A:\cT_2(\mathcal{H})\mapsto \cT_2(\mathcal{H})$ is the operator of left multiplication by $A$, is a faithful representation, turning $(\pi(\cB(\cH)),\cT_2(\cH),\hphantom{n}^*,\cT_2(\cH)_+)$ into a standard form, where $\hphantom{n}^*$ is the usual adjoint, and $T_2(\cH)_+$ is the space of non-negative Hilbert Schmidt operators on $\cH$. As already discussed in \Cref{statenormaldef} any positive, normal functional $\omega$ on $\cB(\cH)$ can be associated with a trace-class operator $\rho$ so that for any $A\in\cB(\cH)$,
	\begin{align}
		\omega(A)=\tr(\rho A)=\tr(\sqrt{\rho}A \sqrt{\rho})=\langle \sqrt{\rho},L_A(\sqrt{\rho})\rangle_{HS},
	\end{align}
	where $\langle.,.\rangle_{HS}$ denotes the Hilbert-Schmidt inner product on $\cT_2(\cH)$. Hence, one identifies $\Omega_{\omega}$ with $\sqrt{\rho}$. Then for any two positive functionals $\omega,\nu$ with associated positive, trace-class operators $\rho,\sigma$:
	\begin{align}
		\Delta_{\rho|\sigma}(A):=\Delta_{\omega|\nu}(A)=\rho A\sigma^{-1}.
	\end{align}
\end{example}		
\subsection{Araki's relative entropy}

In order to rigorously prove \Cref{theinequality} we need to introduce Araki's relative entropy \cite{Araki1976,Ohya}. For any two positive normal functionals on a von Neumann algebra $\mathfrak{V}$, we denote by $\mu_{\nu|\omega}$ the spectral measure for $-\ln \Delta_{\nu|\omega}$ with respect to the state $\omega$. This means that it is the only probability measure on the spectrum of $\Delta_{\nu|\omega}$ such that for any bounded measurable function $f$ on $\spec(\Delta_{\nu|\omega})$,
\begin{align}
	\langle \Omega_\omega,f(-\ln\Delta_{\nu|\omega})\Omega_\omega\rangle=\int_{\spec(\Delta_{\nu|\omega})}f(x)\mu_{\mu|\nu}(dx).
\end{align}
Then Araki's \textit{relative entropy} of $\omega$ with respect to $\nu$ is defined by
\begin{align}
	\operatorname{Ent}(\omega|\nu):=\left\{
	\begin{aligned}
		&-\langle\Omega_\omega, \ln(\Delta_{\nu|\omega})\Omega_\omega\rangle=\int_{\spec(\Delta_{\nu|\omega})}x \mu_{\nu|\omega}(dx)\quad \omega<<\nu\\&+\infty\qquad\text{ otherwise}.
	\end{aligned}
	\right.
\end{align}
where $\omega<<\nu$ means that $\omega$ is normal with respect to $\nu$, i.e. that $P_\omega\le P_\nu$. Roughly speaking, the above quantity characterizes the ``distance" between two positive normal functionals. As advertised at the beginning of this section, Araki's relative entropy reduces to the classical Kullback-Leibler divergence in the case of classical probability distributions absolutely continuous with respect to a given measure, and to Umegaki's quantum relative entropy in the case of density operators $\rho,\sigma$ on a Hilbert space $\cH$ such that $\supp\rho\subset \supp\sigma$, which makes it very attractive from an abstract point of vue:
\begin{example}[Classical probability theory, continued]
	Araki's relative entropy reduces to the classical Kullback-Leibler divergence
	\begin{align}\label{KLdiv}
		D_{\rm{KL}}(f||g):=\operatorname{Ent}(\omega_f||\omega_g)=\int_\Omega f(x)\ln\frac{f(x)}{g(x)}\mu(dx)
	\end{align}	
	for states $\omega_f,\omega_g$ with associated integrable functions $f,g$ defined as in \Cref{statesclass}.
\end{example}
\begin{example}[Quantum systems, continued]
	It also reduces to Umegaki's quantum relative entropy 
	\begin{align}
		D(\rho||\sigma):=\operatorname{Ent}(\omega|\nu)=\tr(\rho(\ln \rho-\ln\sigma))
	\end{align}
	for faithful states $\omega, \nu$ on $\cB(\cH)$ with associated positive density operators $\rho,\sigma$ such that $\supp \rho\subseteq \supp \sigma$.
\end{example}
Moreover, Araki's relative entropy satisfies the following useful properties:
\begin{theorem}[\cite{Ohya} Theorem 5.20]
	Let $\mathfrak{V}$, $\mathfrak{V}_1$ and $\mathfrak{V}_2$ be three von Neumann algebras, let $\omega$ be a normal state on $\mathfrak{V}$, $\omega_1,\sigma_1$ be normal states on $\mathfrak{V}_1$ and $\omega_2,\sigma_2$ be normal states on $\mathfrak{V}_2$. Then\footnote{For an exposition of tensor products of von Neumann algebras and tensor products of states, see e.g. \cite{Kadison}}:
	\begin{align}
		\operatorname{Ent}(\omega||\omega)&=0\\
		\operatorname{Ent}(\omega_{1}\otimes \omega_2||\sigma_1\otimes\sigma_2)&= \operatorname{Ent}(\omega_1||\sigma_1)+ \operatorname{Ent}(\omega_2||\sigma_2)
	\end{align}
	
\end{theorem}
The following theorem by Uhlmann is at the heart of our proof:
\begin{theorem}[Uhlmann's monotonicity] Let $\mathcal{V}$ and $\mathcal{W}$ be two von Neumann algebras, and let $\omega,\nu$ be normal states on $\mathcal{W}$. Let $\Phi:\mathcal{V}\rightarrow \mathcal{W}$ be a Schwarz map. Then
	\begin{align}
		\operatorname{Ent}(\omega\circ\Phi||\nu\circ\Phi)\le \operatorname{Ent}(\omega||\nu) 
	\end{align}
\end{theorem}
\section{Proof of the inequality \cref{theinequality}}\label{proofineqpos}
In this section we rigorously prove the inequality \cref{theinequality}. In order to do so, we will use Uhlmann's monotonicity of the entropy power under Schwarz mapping in the context of general von Neumann algebras. 
Using this theorem, and with the notations of \Cref{concavity} we first show the following inequality:
\begin{align}\label{entropyineq}
	\operatorname{Ent}( \omega_{\rho \ast{g_{t/2}}} || \omega_{\rho_{R_j}^{(\theta\alpha)}\ast {g_{R_j,t/2}^{(\theta\beta)}}}) \leq D(\rho||\rho_{R_j}^{(\theta\alpha)})+ D_{\rm{KL}}(g_{t/2}||g_{R_j,t/2}^{(\theta\beta)})
\end{align}
This follows from Uhlmann's monotonicity theorem if we can prove that for any $\rho\in\cD(\cH)$ and any positive $g\in L^1(\RR^{2n},\cB(\RR^{2n}),\lambda)$, with $\|g\|_{1}=1$,
\begin{align}
	\omega_{\rho*g} = (\omega_\rho\otimes \omega_g) \circ \Phi
\end{align}
where $\Phi$ is a unital, completely positive map from $\cB(\cH)$ to$\cB(\cH)\otimes L^{\infty}(\RR^{2n},\cB(\RR^{2n}),\lambda)$, with $\lambda$ being the Lebesgue measure on ${\mathbb{R}}^{2n}$.
 
Let us see how $\omega_{\rho}\otimes \omega_{g}$ and $\omega_{\rho*g}$ are defined. $\omega_{\rho}\otimes \omega_{g}$ acts on tensor product elements of $\mathcal B(\mathcal H)\otimes L^\infty(\RR^{2n}) $ by 
\begin{align}
	(\omega_{\rho}\otimes \omega_{g}) (A \otimes f)=\tr(\rho A) \times \int_{\RR^{2n}} g(z)f(z)dz,\quad A\in\cB(\cH), f\in L^{\infty}(\RR^{2n})
\end{align}
and extension by continuity \cite{Kadison}. But $\mathcal B(\mathcal H)\otimes L^\infty(\RR^{2n}) $ can be identified with the set $L^\infty \big(\RR^{2n},\mathcal B(\mathcal H)\big)$ of measurable functions from $\RR^{2n}$ to $\cB(\cH)$, in which case an element of it is a map $z\mapsto A_z$ and $\omega_\rho\otimes \omega_g$ acts as
\begin{align}
	(\omega_\rho\otimes \omega_g)(z\mapsto A_z) = \int_{\RR^{2n}} \tr(\rho A_z) \,g(z)\,\mathrm dz
\end{align}
(which of course is consistent with the above if $A_z=f(z) A$). On the other hand, $\omega_{\rho \ast g}$ acts as
\begin{align}
	\omega_{\rho\ast g}(A) = \int_{\RR^{2n}} g(z) \,\tr(\rho V(-z) A V(z))\, \mathrm dz,
\end{align}
so to prove our statement it is enough to show that the map $\Phi$ from $\mathcal B (\mathcal H)$ to $L^\infty\big(\RR^{2n}, \mathcal B(\mathcal H)\big) $ defined as
\begin{align}
	\Phi : A \mapsto (z\mapsto V(-z) A V(z))
\end{align}
is unital and completely positive. The first requirement is obvious. To show complete positivity, consider a positive matrix $(A_{i,j})_{i,j=1}^N$ of elements of $\mathcal B(\mathcal H) $, in the sense that for any $(\varphi_i)_{i=1}^N$ with each $\varphi_i\in\mathcal H$, one has 
\[\sum_{i,j} \langle \varphi_i, A_{i,j}\varphi_j \rangle \geq 0.\]
We then prove that $(\Phi(A_{i,j}))_{i,j=1}^N$ is positive when acting on $L^2(\RR,\mathcal H)$: Consider $\psi_i=(z\mapsto \psi_i(z))\in L^2(\RR,\mathcal H)$, therefore
\begin{align*}
	\sum_{i,j} \langle \psi_i, \Phi(A_{i,j})\psi_j \rangle = \sum_{i,j} \int \langle \psi_i(z), V_z^*A_{i,j}V_z \psi_j(z)\rangle \, \mathrm d z = \int \sum_{i,j} \langle V_z \psi_i(z), A_{i,j}V_z \psi_j(z)\rangle  \, \mathrm d z \geq 0.
\end{align*}
This concludes the proof of \Cref{entropyineq}. However, recall that we proved in \Cref{astast} that
$
\rho_{P_j}^{(\theta \alpha)}\ast g^{(\theta \beta)}_{P_j,t/2} =(\rho_t)_{P_j}^{\left(\theta(\alpha + \beta)\right)}.$ Hence, \Cref{entropyineq} can be rewritten as:
\begin{align}\label{entropyineq2}
	\operatorname{Ent}( \omega_{\rho \ast{g_{t/2}}} || \omega_{\rho_{R_j}\ast g_{R_j,t/2}^{(\theta(\alpha+\beta))}}) \leq D(\rho||\rho_{R_j}^{(\theta\alpha)})+ D_{\rm{KL}}(g_{t/2}||g_{R_j,t/2}^{(\theta\beta)}).
\end{align}

\noindent
In order to obtain the inequality (\ref{theinequality}), we use the following lemma which can be proved by simply using Taylor expansion to second order.
\begin{lemma}
	Let $f,h:\RR\rightarrow \RR$ be two twice differentiable functions such that $h(0)=f(0)=h'(0)=f'(0)=0$ and $f\le h$. Hence, $f''(0)\le h''(0)$.
\end{lemma}	

Take $f(\theta):=\operatorname{Ent} ( \omega_{\rho \ast{g_{t/2}}} || \omega_{\rho_{R_j}\ast g_{R_j,t/2}^{(\theta(\alpha+\beta))}}) $ and $h(\theta) := D(\rho||\rho_{R_j}^{(\theta\alpha)})+ D_{\rm{KL}}(g_{t/2}||g_{R_j,t/2}^{(\theta\beta)})$. These two functions satisfy the requirements of last lemma, and therefore $f''(0)\le g''(0)$. Summing over $j=1,...,2n$ we finally get
\begin{align}
	(\alpha+\beta)^2J(\rho\ast g_{t/2})\le \alpha^2J(\rho)+\beta^2 J_{{cl}}(g_{t/2}),
\end{align}	
Where for any positive, differentiable probability density function $f$, $J_{{cl}}(f)$ is its (classical) Fisher information, defined through \Cref{cl-FI}. Hence, the inequality (\ref{theinequality}) follows.

\bibliography{librarytouse}

\begin{thebibliography}{10}

\bibitem{Araki1976}
H.~Araki.
\newblock {Relative Entropy of States of von Neumann Algebras}.
\newblock {\em Publ. RIMS, Kyoto Univ.}, 11, 1976.

\bibitem{BBG11}
D.~Bakry, F.~Bolley, and I.~Gentil.
\newblock {Around Nash Inequalities}.
\newblock In {\em {Journ{\'e}es {\'e}quations aux d{\'e}riv{\'e}es partielles
  (2010)}}, page
  http://www.proba.jussieu.fr/pageperso/bolley/Nash.PortDAlbret.pdf, 2011.

\bibitem{BK2016}
S.~Beigi and C.~King.
\newblock Hypercontractivity and the logarithmic {S}obolev inequality for the
  completely bounded norm.
\newblock {\em Journal of Mathematical Physics}, 57(1), 2016.

\bibitem{Blachman1965}
N.~M. Blachman.
\newblock {The Convolution Inequality for Entropy Powers}.
\newblock {\em IEEE Transactions on Information Theory}, 11(2), 1965.

\bibitem{Carbone2008}
R.~Carbone and E.~Sasso.
\newblock {Hypercontractivity for a quantum Ornstein-Uhlenbeck semigroup}.
\newblock {\em Probability Theory and Related Fields}, 140(3-4), 2008.

\bibitem{CL93}
E.~A. Carlen and M.~Loss.
\newblock Sharp constant in {N}ash's inequality.
\newblock {\em International Mathematics Research Notices}, (7), 1993.

\bibitem{Chafai2005}
D.~Chafai.
\newblock {In{\'{e}}galit{\'{e}}s de Poincar{\'{e}} et de Gross pour les
  mesures de Bernoulli, de Poisson, et de Gauss}, accessible in
  https://hal.archives-ouvertes.fr/hal-00012428v2.
\newblock Technical report, 2005.

\bibitem{Costa1985}
M.~H.~M. Costa.
\newblock {A New Entropy Power Inequality}.
\newblock {\em IEEE Transactions on Information Theory}, 31(6), 1985.

\bibitem{Davies1993}
E.~B. Davies and J.~M. Lindsay.
\newblock {Non-commutative symmetric Markov semigroups}.
\newblock {\em Mathematische Zeitschrift}, 210(1), 1992.

\bibitem{DMGH15}
G.~De~Palma, A.~Mari, V.~Giovannetti, and A.~S. Holevo.
\newblock Normal form decomposition for {G}aussian-to-{G}aussian
  superoperators.
\newblock {\em Journal of Mathematical Physics}, 56(5), 2015.

\bibitem{Dembo1989}
A.~Dembo.
\newblock {Simple proof of the concavity of the entropy power with respect to
  added Gaussian noise}.
\newblock {\em IEEE Transactions on Information Theory}, 35(4), 1989.

\bibitem{Dembo1991}
A.~Dembo, T.~M. Cover, and J.~A. Thomas.
\newblock {Information Theoretic Inequalities}.
\newblock {\em IEEE Transactions on Information Theory}, 37(6), 1991.

\bibitem{Diaconis1996}
P.~Diaconis and L.~Saloff-Coste.
\newblock {Nash inequalities for finite Markov chains}.
\newblock {\em Journal of Theoretical Probability}, 9(2), 1996.

\bibitem{Gross1975}
L.~Gross.
\newblock {Hypercontractivity and logarithmic Sobolev inequalities for the
  Clifford-Dirichlet form}.
\newblock {\em Duke Mathematical Journal}, 42(3), 1975.

\bibitem{Gross}
L.~Gross.
\newblock Logarithmic {S}obolev inequalities.
\newblock {\em American Journal of Mathematics}, 97(4), 1975.

\bibitem{H94}
M.~J.~W. Hall.
\newblock Gaussian noise and quantum-optical communication.
\newblock {\em Phys. Rev. A}, 50:3295--3303, Oct 1994.

\bibitem{H00}
M.~J.~W. Hall.
\newblock Quantum properties of classical {F}isher information.
\newblock {\em Phys. Rev. A}, 62(1), 2000.

\bibitem{Holevo1996}
A.~S. Holevo.
\newblock {Covariant quantum Markovian evolutions}.
\newblock {\em Journal of Mathematical Physics}, 37(4), 1996.

\bibitem{Holevo2011}
A.~S. Holevo.
\newblock {\em {Probabilistic and Statistical Aspects of Quantum Theory; 2nd
  ed.}}
\newblock Publications of the Scuola Normale Superiore Monographs. Springer,
  2011.

\bibitem{Holevo2012}
A.~S. Holevo.
\newblock {\em {Quantum Systems, Channels, Information, A Mathematical
  Introduction}}.
\newblock De Gruyter, 2012.

\bibitem{koenig2016}
S.~Huber, R.~K\"{o}nig, and A.~Vershynina.
\newblock {Geometric inequalities from phase space translations}.
\newblock {\em arXiv:1606.08603v1}, 2016.

\bibitem{Jaksicb}
V.~Jak$\check{s}$i\'{c}, Y.~Ogata, C.-A. Pillet, and R.~Seiringer.
\newblock Quantum hypothesis testing and non-equilibrium statistical mechanics.
\newblock {\em Reviews in Mathematical Physics}, 24(06), 2012.

\bibitem{Kadison}
R.~Kadison and J.~Ringrose.
\newblock {\em Fundamentals of the Theory of Operator Algebras: Advanced
  theory}.
\newblock American Mathematical Society, 1997.

\bibitem{KT15}
M.~Kastoryano and K.~Temme.
\newblock Non-commutative {N}ash inequalities.
\newblock {\em Journal of Mathematical Physics}, 57(1), 2016.

\bibitem{KT12}
M.~J. Kastoryano and K.~Temme.
\newblock Quantum logarithmic {S}obolev inequalities and rapid mixing.
\newblock {\em Journal of Mathematical Physics}, 54, 2013.

\bibitem{KKW16}
M.~Keyl, J.~Kiukas, and R.~F. Werner.
\newblock Schwartz operators.
\newblock {\em Reviews in Mathematical Physics}, 28(03):1630001, 2016.

\bibitem{Koenig2015}
R.~K\"{o}nig.
\newblock {The conditional entropy power inequality for Gaussian quantum
  states}.
\newblock {\em Journal of Mathematical Physics}, 56(2), 2015.

\bibitem{Koenig2012}
R.~K\"{o}nig and G.~Smith.
\newblock {The entropy power inequality for quantum systems}.
\newblock {\em IEEE Transactions on Information Theory}, 60(3), 2014.

\bibitem{K72}
A.~{Kossakowski}.
\newblock {On quantum statistical mechanics of non-Hamiltonian systems}.
\newblock {\em Reports on Mathematical Physics}, 3, Dec. 1972.

\bibitem{L79}
G.~Lindblad.
\newblock {Gaussian quantum stochastic processes on the CCR algebra}.
\newblock {\em Journal of Mathematical Physics}, 20(10):2081--2087, oct 1979.

\bibitem{M2012}
A.~Montanaro.
\newblock Some applications of hypercontractive inequalities in quantum
  information theory.
\newblock {\em Journal of Mathematical Physics}, 53(12), 2012.

\bibitem{MFW16}
A.~{M{\"u}ller-Hermes}, D.~{Stilck Fran{\c c}a}, and M.~M. {Wolf}.
\newblock {Entropy production of doubly stochastic quantum channels}.
\newblock {\em Journal of Mathematical Physics}, 57(2):022203, Feb. 2016.

\bibitem{Nash}
J.~Nash.
\newblock Continuity of solutions of parabolic and elliptic equations.
\newblock {\em American Journal of Mathematics}, 80(4), 1958.

\bibitem{Nelson}
E.~Nelson.
\newblock A quartic interaction in two dimensions.
\newblock {\em Mathematical Theory of Elementary Particles}, 1966.

\bibitem{NC}
M.~A. Nielsen and I.~L. Chuang.
\newblock {\em Quantum Computation and Quantum Information}.
\newblock Cambridge University Press, 2011.

\bibitem{Ohya}
M.~Ohya and D.~Petz.
\newblock {\em {Quantum Entropy and Its Use}}.
\newblock Springer-Verlag, 1993.

\bibitem{Reedsimon}
M.~Reed and B.~Simon.
\newblock {\em {Methods of Modern Mathematical Physics vol. I Functional
  analysis}}.
\newblock Princeton University Academic Press, 1972.

\bibitem{Stam1959}
A.~Stam.
\newblock {Some inequalities satisfied by the quantities of information of
  Fisher and Shannon}.
\newblock {\em Information and Control}, 2, 1959.

\bibitem{Takesaki}
M.~Takesaki.
\newblock {\em Theory of operator algebras. {I}}, volume 124 of {\em
  Encyclopaedia of Mathematical Sciences}.
\newblock Springer-Verlag, 2002.

\bibitem{Aspects2003}
M.~Takesaki.
\newblock {\em {Theory of Operator Algebras II}}, volume 125 of {\em
  Encyclopaedia of Mathematical Sciences}.
\newblock Springer-Verlag, 2003.

\bibitem{TOVPV11}
K.~Temme, T.~Osborne, K.~Vollbrecht, D.~Poulin, and F.~Verstraete.
\newblock Quantum {M}etropolis sampling.
\newblock {\em Nature}, 2011.

\bibitem{Toscani}
G.~Toscani.
\newblock An information-theoretic proof of {N}ash's inequality.
\newblock {\em Rendiconti Lincei - Matematica e Applicazioni}, 24(1), 2013.

\bibitem{Umegaki1962}
H.~Umegaki.
\newblock {Conditional expectation in an operator algebra IV. Entropy and
  information}.
\newblock {\em Kodai Mathematical Seminar Reports}, 14(2), 1962.

\bibitem{Villani2000}
C.~Villani.
\newblock {A short proof of the ``concavity of entropy power''}.
\newblock {\em IEEE Transactions on Information Theory}, 46(4), 2000.

\bibitem{Weissler1978}
F.~B. Weissler.
\newblock {Logarithmic Sobolev Inequalities for the Heat-Diffusion Semigroup}.
\newblock {\em Transactions of the American Mathematical Society}, 237, 1978.

\bibitem{Werner1984}
R.~Werner.
\newblock {Quantum harmonic analysis on phase space}.
\newblock {\em Journal of Mathematical Physics}, 25(5), 1984.

\end{thebibliography}

\end{document}